\documentclass[journal,twoside]{IEEEtranTCOM}
\usepackage{mathrsfs}
\usepackage{latexsym}
\usepackage{graphicx}
\usepackage{epsfig}
\usepackage{subfigure}
\usepackage{array}
\usepackage{amsmath}
\usepackage{amssymb}
\usepackage{color, soul}
\usepackage{amsthm}
\usepackage{enumerate}
\usepackage{algpseudocode}
\usepackage{algorithm}
\usepackage{bm}
\usepackage{amssymb}
\usepackage{balance}
\usepackage{verbatim}
\usepackage{epstopdf}
\usepackage{setspace}
\usepackage{xcite}
\usepackage{xr-hyper}
\usepackage{algpseudocode}
\usepackage{algorithm}
\usepackage{booktabs}
\usepackage{tabularx} 
\usepackage{ragged2e} 
\usepackage{booktabs} 
\usepackage{balance}

\newtheorem{thm}{Theorem}

\newtheorem{lemma}{Lemma}

\theoremstyle{definition}

\normalsize
\begin{document}

	\title{ Dynamic Agile Reconfigurable Intelligent Surface Antenna (DARISA) MIMO: DoF Analysis and Effective DoF Optimization
		\thanks{Jiale Bai and Hui-Ming Wang are with the School of Information and Communications Engineering, Xi'an Jiaotong University,
			Xi'an 710049, China (e-mail: bjl19970954@stu.xjtu.edu.cn; xjbswhm@gmail.com).
		}
		\thanks{ Liang Jin is with the PLA Strategic Support Force Information Engineering University, Zhengzhou 450001, China
			(e-mail: liangjin@263.net).
		}
		\author{Jiale Bai, Hui-Ming Wang, \emph{Senior Member, IEEE}, and Liang Jin
		}
	}
	\maketitle
	\begin{abstract}
			In this paper, we propose a dynamic agile reconfigurable intelligent surface antenna (DARISA) array  integrated into multi-input multi-output (MIMO) transceivers. Each DARISA comprises a number of metasurface elements activated simultaneously via a parallel feed network. The proposed system enables rapid and intelligent phase response adjustments for each metasurface element within a single symbol duration, facilitating a dynamic agile adjustment of phase response (DAAPR) strategy. By analyzing the theoretical degrees of freedom (DoF) of the DARISA MIMO system under the DAAPR framework, we derive an explicit relationship between DoF and critical system parameters, including agility frequentness (i.e., the number of phase adjustments of metasurface elements during one symbol period), cluster angular spread of  wireless channels, DARISA array size, and the number of transmit/receive DARISAs. The DoF result reveals a	significant conclusion: when the number of receive DARISAs is smaller than that of transmit DARISAs, the DAAPR strategy of the DARISA MIMO enhances the overall system DoF. Furthermore, relying on DoF alone to measure channel capacity is insufficient, so we analyze the effective DoF (EDoF) that reflects the impacts of the DoF and channel matrix singular value distribution on capacity. We show channel capacity monotonically increases with EDoF, and optimize the agile phase responses of metasurface elements by using fractional programming (FP) and semidefinite relaxation (SDR) algorithms to maximize the EDoF. Simulations validate the theoretical DoF gains and reveal that increasing agility frequentness, metasurface element density, and phase quantization accuracy can enhance the EDoF. Additionally, densely deployed elements can compensate for the loss in communication performance caused by lower phase quantization accuracy.
	\end{abstract}	
	
	
	\section{Introduction}
	
	With the rapid development of wireless communications, increasing data rates and capacity remains a critical challenge. The massive MIMO is an effective technology for delivering high data traffic. However, its high power consumption, low energy efficiency, and substantial hardware costs exacerbate the overall design complexity of communication systems. In this context, emerging reconfigurable intelligent surface (RIS) technology has been proposed as a promising solution for improving data rates and reducing energy consumption \cite{Wu1}. Specifically, RIS is an innovative artificial electromagnetic meta-surface that integrates numerous controllable elements in a compact space, enabling intelligent control over the electromagnetic characteristics of each element through digital programming.
	Research efforts have focused on deploying RIS as passive reflectors to assist wireless communications \cite{Huang1}-\cite{Pan-20a}. By optimizing the phase shifts of these passive RIS elements, reflected signals can constructively combine with direct channel links at receivers, improving  communication performance such as channel capacity, energy efficiency, and physical layer security \cite{Dong-20}-\cite{Bai-22}. However, passive RIS-aided channels are inherently constrained by a multiplicative fading attenuation, severely limiting received signal power. To address this, active RIS (ARIS) architectures incorporating signal amplification have been proposed \cite{Long-21}-\cite{Lv-23}. Nonetheless, active RIS suffers from high power consumption and hardware costs.
	
	Recent advancements have expanded RIS applications beyond conventional reflection-based architectures. For instance, RIS technology has been embedded directly into antenna structures, termed holographic MIMO surfaces (HMIMOS) \cite{Huang-20}, enabling real-time reconfiguration of both transmitted and received electromagnetic waveforms.

	HMIMOS integrates densely packed metamaterial elements with reconfigurable processing networks, forming continuous antenna apertures that enhance spatial degrees of freedom (DoF) and channel capacity. Various channel models \cite{Pizzo1}-\cite{Sun} have been established, with extensive analysis of their derived spatial DoF. Pizzo et al. \cite{Pizzo1} employed a Fourier plane-wave representation channel model grounded in electromagnetic theory to demonstrate that spatial DoF scaled proportionally with HMIMOS aperture size in isotropic scattering environments. Extending this framework, Pizzo et al. \cite{Pizzo3} revealed that non-isotropic scattering environments constrain achievable DoFs. Yuan et al. \cite{Yuan1} further analyzed inter-element coupling effects, revealing that these coupling effects at receivers often reduced the eigenvalue magnitudes of the channel matrix while increasing the spatial DoF. Furthermore, Wei et al. \cite{Wei} confirmed that minimizing element spacing under fixed array dimensions intensified signal correlation, thereby degrading spectral efficiency in multi-user HMIMOS deployments. An et al. \cite{An} theoretically established that expanding the element count in stacked-HMIMO architectures substantially elevates channel capacity bounds. Collectively, these studies show that HMIMOS achieve substantial spatial DoFs, which are determined by both the HMIMOS size and the scattering environment, thereby necessitating extensive RF chains to fully exploit the available DoFs.
	
	However, existing studies on HMIMOS \cite{Pizzo1}-\cite{Sun} neglect the impact of real-time dynamic phase manipulation capabilities (dynamic agile adjustment) in metasurface elements
	 in the time domain. A critical feature of metasurfaces is their dynamic agile adjustment of phase response, where each element independently reconfigures its phase response at a nanosecond-scale rate. For instance, a demonstrated prototype \cite{Sleasman} employs positive intrinsic-negative (PIN) diodes with 2-ns response times for phase tuning. This rapid response time is significantly shorter than one expected symbol duration, enabling the dynamic agile \emph{multiple adjustments} of phase responses of all metamaterial elements within one symbol transmission period. Consequently, this dynamic agile approach increases the number of available samples from multiple observations of one symbol across joint space-time domains without expanding the signal bandwidth, which may potentially further enhance the utilization of DoF and improve communication performance.
	
	On the other hand, a fundamental principle of MIMO communications states that channel capacity is jointly determined by both spatial DoF (the rank of the channel matrix) and the distribution of its singular values \cite{Tse}. However, existing studies \cite{Pizzo1}-\cite{Sun}, which focus on channel DoF perspective, have overlooked the impact of singular value distributions on capacity. Consequently, relying solely on DoF to measure channel capacity is not sufficient. To address this limitation, the concept of ``effective degrees of freedom (EDoF)" was developed in \cite{Shiu}-\cite{Yuan2} to provide a more accurate reflection of channel capacity. Specifically, the EDoF of a MIMO system indicates the equivalent quantity and quality of its independent SISO system, providing a convenient single metric for estimating the performance of a MIMO system. Unlike conventional DoF, the EDoF is directly related to the slope of spectral efficiency (channel capacity at a single frequency) \cite{Verdu}, thus making it more directly related to channel capacity. However, so far the EDoF of an HMIMOS system with dynamic adjustment capability has not been investigated.
	
	To address these gaps, we develop a novel dynamic agile reconfigurable intelligent surface antenna (DARISA) MIMO transmission system, where transceivers employ multiple DARISAs.
	  The architecture of a DARISA is depicted in Fig. 1, where each DARISA consists of phase control circuits, numerous metamaterial elements, a power distribution network, and feed ports. The phase control circuits of each metamaterial element realize phase adjustment by adjusting the control voltage to a varactor diode, which is controlled by a DARISA controller. All metamaterial elements are activated simultaneously and parallelly by a coaxial feed network. The RF power is distributed to each element uniformly by a power distribution network. This architecture enhances inter-port isolation while maintaining low implementation complexity.

The proposed DARISA architecture is different from existing technologies such as conventional MIMO relay \cite{Wang-10}\cite{Bai-18}, HMIMOS \cite{Pizzo1}-\cite{Sun} and dynamic metasurface antenna (DMA) \cite{Hunt}-\cite{Shlezinger}. Conventional relays employ independent RF chains per antenna as distinct nodes to extend coverage, necessitating full digital receive-process-retransmit chains and passively forwarding signals through two-hop links. In contrast, DARISA functions as a digital-analog hybrid transceiver, utilizing minimal digital RF chains to govern large-scale analog metasurface elements. This architecture establishes direct communication links and actively optimizes channel matrices, operating fundamentally as a channel-optimizing transceiver rather than a signal repeater. Conventional HMIMOS relies on densely packed antenna arrays and necessitates extensive RF chains to maximize spatial DoF, resulting in extremely high hardware complexity and power consumption. Furthermore, conventional HMIMOS overlooked the ability of dynamic agile adjustment of metasurfaces elements in the time domain. The DMA employs microstrip waveguides to sequentially activate elements via guided signal propagation, whereas DARISA leverages a parallel coaxial feed network to activate all metasurface elements simultaneously. Moreover, while the term ``dynamic" in DMA denotes space-domain phase adjustments, DARISA achieves dynamic agile phase reconfiguration in the time domain.

	\begin{figure}[!t]
		\centering
		{\includegraphics[width=3.5 in]{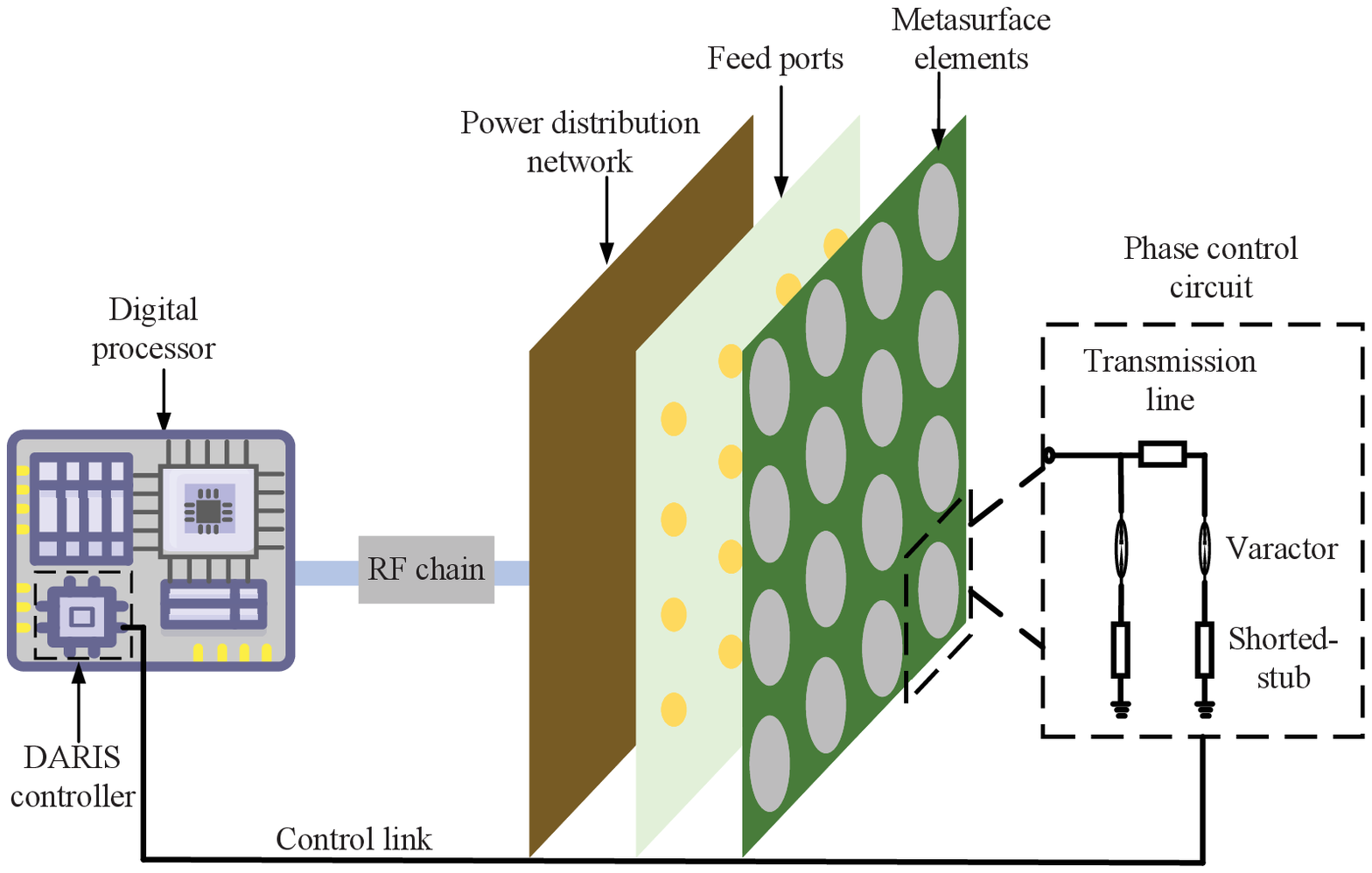}}
		\caption{Architecture of a DARISA.}
	\end{figure}
	The most significant feature of the DARISA lies in the term ``agile'', which utilizes the nanoseconds order of phase adjustment capability of the metesurface. This work investigates whether this agile capability could impact the DoF when both transmitter and receiver are equipped with DARISA arrays. With a spatial multi-clusters channel model \cite{Poon}, which effectively characterizes the impact of spatial features such as angular spread on channels, a dynamic agile adjustment of phase response (DAAPR) strategy is proposed. By dynamically reconfiguring metasurface phase responses within a single symbol duration without bandwidth expansion, we construct an equivalent extended spatial-temporal channel model that facilitates multi-dimensional signal observations from joint time and space domains, allowing for full utilization of the potential DoF. Furthermore, we investigate the EDoF metric for the DARISA MIMO system and optimize the agile phase responses of elements to maximize EDoF.
		
		The main contributions and innovations of this work are summarized as follows:
		
		1) We leverage the real-time controllability of metasurfaces to develop ``DARISA MIMO" at the transceiver. The communication potential of DARISA MIMO is exploited in the joint space-time domain by constructing an equivalent extended spatial-temporal channel through the dynamic agile phase responses adjustment of metasurface elements.  Additionally, since each DARISA requires only a single RF chain, the proposed DAAPR strategy significantly reduces hardware costs and power consumption, which are high in HMIMOS systems due to the need for numerous RF chains.
		
		2) We analyze the DoF of the DARISA MIMO system in both isotropic and non-isotropic scattering environments. An explicit relationship between DoF and critical system parameters, including agility frequentness, cluster angular spread, transceiver DARISA array sizes, and the number of transceiver DARISAs, is derived. Our findings indicate that the DoF can be improved by appropriately designing the agile adjustment of the phase response of metasurface elements when the number of receive-DARISA is less than that of transmit-DARISA.
		
		3)  To more accurately measure the channel capacity, we analyze the effective DoF (EDoF) that reflects the impacts of the DoF and channel matrix singular value distribution on capacity. We show that the channel capacity monotonically increases with EDoF, and we optimize the phase responses of metasurface elements to maximize the EDoF under the DAAPR strategy. To solve this optimization problem, fractional programming and semidefinite relaxation (SDR) algorithms are proposed. Simulation results show that dynamic agile adjustment of phase responses improves the EDoF, and deploying a denser number of metasurface elements in a fixed-size array further enhances the EDoF. Additionally, we demonstrate that densely deployed elements can compensate for the loss in communication performance caused by lower phase quantization accuracy.
	
	The remainder of this paper is organized as follows. We introduce the system model and channel model in Section II. In section III, we analyze the DoF of DARISA MIMO. The EDoF is optimized in section IV. Section V shows the simulation results to evaluate the performance of the proposed algorithms. Finally, we conclude the paper in section VI.
	
	\emph{Notations}: 	 For matrix ${\bf A}$, ${\bf A}^{ T}$, ${\bf A}^{ H}$ and ${\bf A}^{*}$ represent the transpose, conjugate transpose and conjugate operations, respectively; ${Tr}({\bf A})$ denotes the trace operation; ${rank}({\bf A})$ is the rank of ${\bf A}$; $\odot$ represents the Hardmard product. For vector ${\bf a}$, $\|{\bf a}\|$ represents the norm operation; $\arg({\bf a})$ denotes taking the phase of ${\bf a}$; ${diag}({\bf a})$ is to transform ${\bf a}$ as a diagonal matrix with diagonal elements in ${\bf a}$; $\mathbb{E}\left \{ \cdot  \right \}$ is statistical expectation; $\jmath=\sqrt{-1}$ is an imaginary unit.

	\section{DARISA System and Channel Model}
	
		In this section, we first describe the DARISA  architecture in detail. Then, a spatial channel model for the DARISA SISO system is introduced. Next, we extend this SISO spatial channel model to a DARISA MIMO system and subsequently present a spatial-temporal channel model.

	\subsection{ DARISA  Architecture }
	The architecture of a DARISA is depicted in Fig. 1, where each DARISA consists of phase control circuits, numerous metamaterial elements, a power distribution network, and feed ports. The phase control circuit of each metamaterial element is composed of a variode and the control biasing  circuit, which can independently adjust the control voltage of the variable diode to realize the phase adjustment of elements. The RF power is distributed to each element uniformly by the power distribution network. All metamaterial elements are activated simultaneously and connected to the RF chain and sequentially to the digital processor through a parallel feed network. These four parts are tightly connected, which ensures a compact antenna architecture while facilitating the easy expansion of the size of the antenna.
	
	\begin{figure}[!t]
		\centerline{\includegraphics[width=3.8 in]{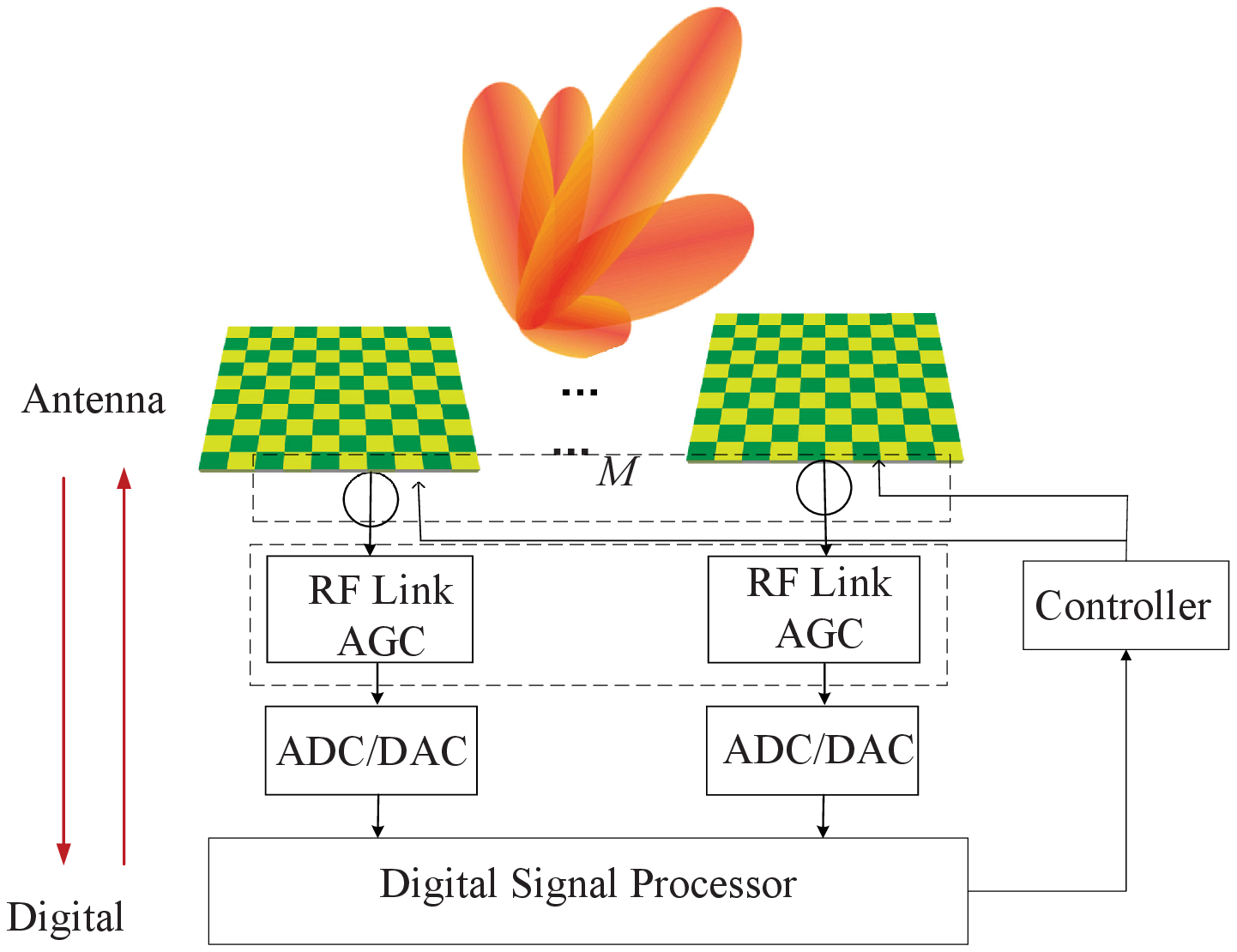}}
		\caption{The architecture of the DARISA transceiver system.}
	\end{figure}
	
	The working principle of this DARISA system is shown in Fig. 2. In receive mode, all metasurface elements capture incoming signals, which are routed to the digital signal processor (DSP) via the RF links and analog-to-digital converter (ADC). The DSP performs channel estimation and signal demodulation, then optimizes the phase configurations of the metasurface elements through the DARISA controller to generate adaptive receive beam patterns. Within each symbol duration, the DSP iteratively generates phase control signals to dynamically adjust metasurface element phases, enabling  different receive beamformers. All received signals are transmitted to the DSP in real time via the RF link, where they undergo joint processing to improve communication performance.
	
In transmit mode, the DSP encodes and digitally precodes the signal, which is then converted by the digital-to-analog converter (DAC) and transmitted via RF links to metasurface elements. The DSP employs a similar dynamic adjustment mechanism to reconfigure metasurface phases in real-time, enabling rapid multi-beamforming for signal transmission.
	
	\subsection{DARISA SISO Spatial Channel Model }
Both the transmitter and receiver equip DARISA consisting of $N_g=N_{g,x}N_{g,y},g\in\left\{r,t\right\}$ metasurface elements, where $N_{g,x}$ and $N_{g,y}$ represent the number of elements per row and per column, respectively. The normalized spacing of adjacent elements is denoted as $\Delta_g$, which is less than or equal to half a wavelength. Thus the horizontal and vertical normalized length of the transmitter and receiver antenna are $D_{g,x}=\Delta_gN_{g,x}$ and $D_{g,y}=\Delta_gN_{g,y},g\in\left\{r,t\right\}$.  The indices of the transmit and receive elements are denoted by $j\in \mathcal{J} \triangleq \left\{1,\cdots,{N}_{t}\right\}$ and $i\in \mathcal{I} \triangleq\left\{1,\cdots,{N}_{r}\right\}$. The normalized location of the $j$-th and $i$-th elements with respect to the origin are expressed as
	\begin{align}
		\notag
		&{\mathbf{t}}_{j}=[t_{x}^{j},t_{y}^{j},t_{z}^{j}]^T={[t_x(j){\Delta }_{t},t_y(j){\Delta }_{t},0]}^{T},j\in \mathcal{J},\\
		\notag
		&{\mathbf{r}}_{i}=[r_{x}^{i},r_{y}^{i},r_{z}^{i}]^T={[r_x(i){\Delta }_{r},r_y(i){\Delta }_{r},0]}^{T},i\in \mathcal{I},
	\end{align}
	where $t_x(j)=\bmod (j-1,{{N}_{t,x}})$ and $r_x(i)=\bmod (i-1,{{N}_{r,x}})$ are horizontal indices of $j$-th transmit element and $i$-th receive element, respectively;  $t_y(j)=\left\lfloor j-1/{{N}_{t,x}} \right\rfloor$ and $r_y(i)=\left\lfloor i-1/{{N}_{r,x}} \right\rfloor$ ar vertical indices of $j$-th transmit element and $i$-th receive element \cite{Wei}, respectively.
	
	We now describe the DARISA SISO spatial channel model. Due to the presence of multiple scatterers in the propagation environment, the wireless channel can be characterized as a typical multipath channel, exemplified by the cluster delay line (CDL) model from the 3GPP protocol TR38.900 \cite{TR}. A flat fading multiple-cluster model which includes $L$ scattering clusters between the transmitter and receiver is considered. Each cluster consists of a superposition of multiple rays that share the same time delay and power within the cluster.

		\begin{figure}[!t]
		\label{SISO-Channel-model}
		\centerline{\includegraphics[width=3.4 in]{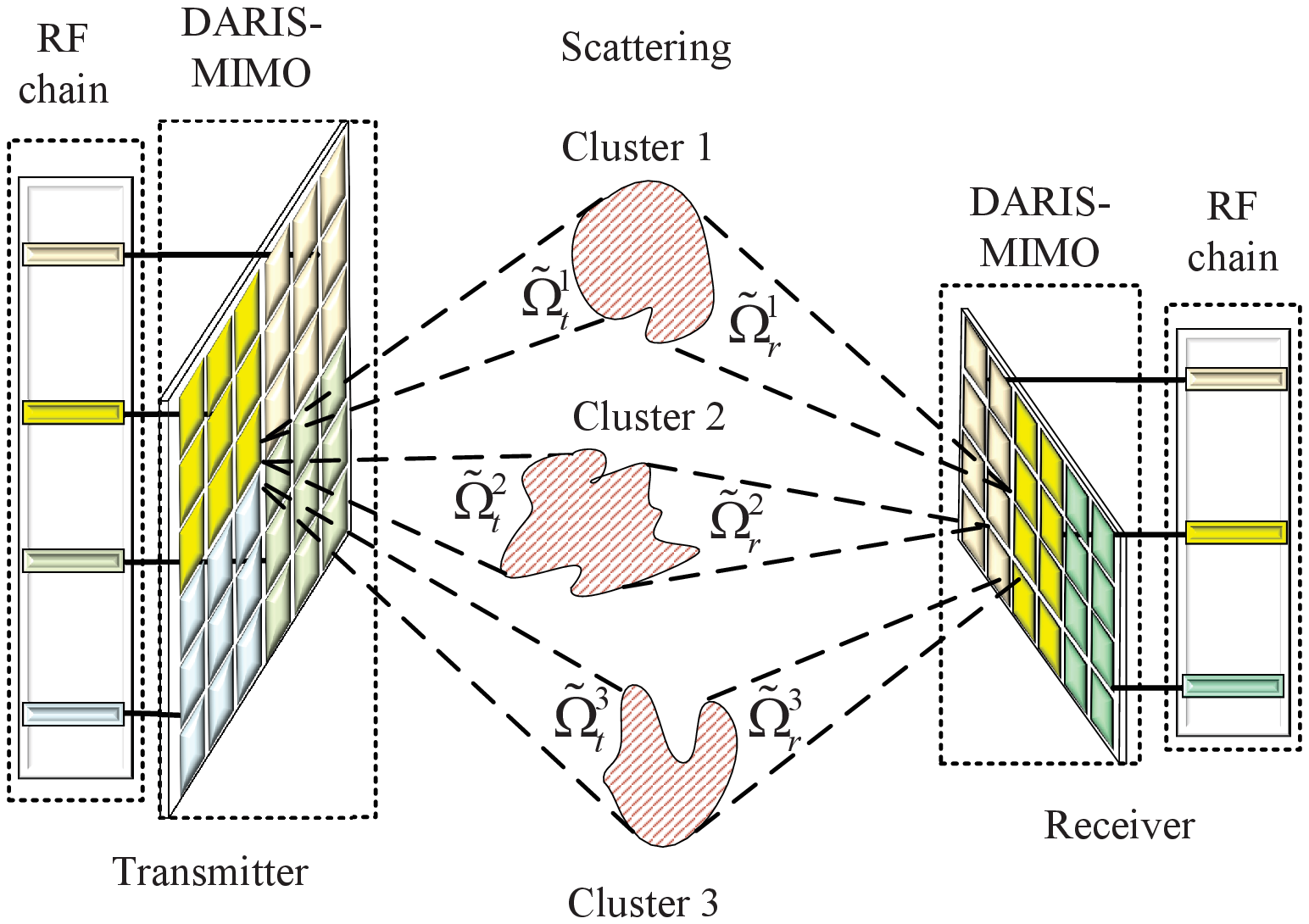}}
		\caption{A DARISA MIMO transmission system.}
	\end{figure}
	
	The center angles of azimuth angle of departure (AoD), zenith angle of departure (ZoD), azimuth angle of arrival (AoA), and zenith angle of arrival (ZoA) in the $l$-th cluster are defined as $\dot{\phi}_{g}^{l}$ and $\dot\theta_{g}^{l}, g\in\left\{r,t\right\},l=1,\cdots,L$, respectively. The angular spreads of AoD, ZoD, AoA, and ZoA in the $l$-th cluster are denoted as $\Delta{\phi}_{g}^{l}$ and $\Delta{\theta}_{g}^{l}$, respectively. Therefore, the angles of AoD($\theta^{l}_{t}$), ZoD($\phi^{l}_{t}$), AoA($\theta^{l}_{r}$), and ZoA($\phi^{l}_{r}$) for any ray in the $l$-th cluster must satisfy the condition
	\begin{align}
		\notag
		&\theta_g^l \in \Omega_{g,\theta}^{l}\triangleq \left[\dot\theta_{g}^{l}- {\Delta{\theta}_{g}^{l}},\dot\theta_{g}^{l}+{\Delta{\theta}_{g}^{l}} \right],\\
		\label{angle-range}
		&\phi_{g}^{l} \in \Omega_{g,\phi}^{l} \triangleq\left[\dot\phi_{g}^{l}- {\Delta{\phi}_{g}^{l}},\dot\phi_{g}^{l}+{\Delta{\phi}_{g}^{l}}\right], g\in\left\{r,t\right\}.
	\end{align}
	where $\Omega_{g,\theta}^{l}$ and $\Omega_{g,\phi}^{l}$ denote the angular ranges for AoD, ZoD, AoA and ZoA in the $l$-th cluster, respectively. To simplify notations, we define the departure and arrival directions along with the angular ranges for the $l$-th cluster, as ${\boldsymbol{\alpha} }_{g}^l\triangleq\left[\theta_g^l,\phi_g^l\right]^T$, and ${{\Omega}}_{g}^{l}\triangleq \Omega_{g,\theta}^{l}\times \Omega_{g,\phi}^{l},g\in\left\{r,t\right\}$, respectively.
	
	Then, the spatial channel response $h_{ij}$ between $i$-th receive and $j$-th transmit element is denoted as \cite{So}
	\begin{align}
		\notag
		h_{ij}=\frac{1}{\sqrt{L}}\sum\limits_{l=1}^{L} \iint_{ {{\Omega}}_{t}^{l} \times {{\Omega}}_{r}^{l} }& q_{r}^{i}{a}_{r}^{i}\left( {\boldsymbol{\alpha} }^l_{r}\right)H_{a}\left( {\boldsymbol{\alpha}}^l_{r}, {\boldsymbol{\alpha} }^l_t\right)\\
		\label{SISO_Continuous}
		& {a}_{t}^{j*}\left({\boldsymbol{\alpha}}^l_t \right)q_{t}^{j}d {\boldsymbol{\alpha}}^l_t d {\boldsymbol{\alpha}}^l_r,
	\end{align}
	where
	\begin{align}
		\notag
		{a}^{i}_{r}\left({\boldsymbol{\alpha}}^l_r\right)&=e^{-\jmath2\pi\left(r_{x}^{i} \cos \phi^l_r\sin \theta^l_r + r_{y}^{i} \sin \phi^l_r\sin \theta^l_r + r_{z}^{i}\cos \theta^l_r\right)},\\
		\notag
		{a}_{t}^{j}\left({\boldsymbol{\alpha} }^l_t\right)&=e ^{-\jmath2\pi\left(t_{x}^{j} \cos \phi^l_t\sin \theta^l_t + t_{y}^{j} \sin \phi^l_t\sin \theta^l_t + t_{z}^{j} \cos \theta^l_t\right)}
	\end{align}
denotes the receive (transmit) response of the $i$-th receive ($j$-th transmit) element to the arrival (departure) direction ${\boldsymbol{\alpha}}^l_r$ (${\boldsymbol{\alpha}}^l_t$) of the $l$-th cluster. The adjustable phase responses of the $j$-th transmit and $i$-th receive elements are given by $q_{t}^{j}=e^{\jmath\varphi_{t}^{j}}$ and $q_{r}^{i}=e^{\jmath\varphi_{r}^{i}}$, respectively. $H_{a}\left({\boldsymbol{\alpha}}^l_r, {\boldsymbol{\alpha} }^l_t\right)$ represents the scattering response of the $l$-th cluster from ${\boldsymbol{\alpha} }^l_t$ to ${\boldsymbol{\alpha}}^l_r$.
	
	Since $h_{ij}$ is bandlimited \cite{Poon}, this channel with continuous angular can be perfectly reconstructed from a discrete angular sequence of samples, provided that the sampling rate satisfies the Nyquist condition \cite{Pizzo1}-\cite{Pizzo3},\cite{Tse}.  When the sampling rate of $h_{ij}$ meets the Nyquist criterion, the number of samples in the departure and arrival directions at the $l$-th cluster are denoted as $d_r^l$ and $d_t^l$, respectively. Since all elements are ultimately connected to a single RF channel in the SISO system, when the transmit and receive antennas consist of ${N}_{t}$ and ${N}_{r}$ elements, the DARISA SISO spatial channel is represented discretely as
	\begin{align}
		\notag
		h& =\sum\limits_{i=1}^{{{N}_{r}}}{\sum\limits_{j=1}^{{{N}_{t}}}}{h_{ij}}\\
		\notag
		&=\frac{1}{\sqrt{L}}\sum\limits_{i=1}^{N_r}\sum\limits_{j=1}^{N_t}\sum\limits_{l=1}^{L} q_{r}^{i}{\bf a}_{r,i}^{lH} {\bf H}_{a}^{l} {\bf a}_{t,j}^{l}q_{t}^{j}\\
		\notag
		&=\frac{1}{\sqrt{L}}\sum\limits_{l=1}^{L} {\bf{q}}_{r}^H\underbrace{{\bf A}_{r}^{lH}{\bf H}_{a}^{l}{\bf A}_{t}^{l}}_{{\bf H}_{w}^{l}}{\bf{q}}_{t}\\
		\label{final_siso_channel}
		&= {\bf{q}}_{r}^H{\bf H}_{w}{\bf{q}}_{t},
	\end{align}
	where $h_{ij}$ in (\ref{SISO_Continuous}) is described by
	\begin{align}
		\notag
		{h_{ij}}=\frac{1}{\sqrt{L}}\sum\limits_{l=1}^{L} \sum\limits_{p_l=1}^{d_r^{l}}\sum\limits_{q_l=1}^{d_t^{l}} q_{r}^{i}{a}_{r}^{i}\left({\boldsymbol{\alpha}}^{p_l}_{r}\right)H_{a}\left({\boldsymbol{\alpha} }^{p_l}_r,{\boldsymbol{\alpha}}^{q_l}_{t}\right)
		{a}_{t}^{j*}\left({\boldsymbol{\alpha}}^{q_l}_{t}\right)q_{t}^{j}.
	\end{align}
	The notations ${a}_{r}^{i}\left({\boldsymbol{\alpha}}^{p_l}_{r}\right)$, ${a}_{t}^{j}\left({\boldsymbol{\alpha}}^{p_l}_{t}\right)$ and $H_{a}\left({\boldsymbol{\alpha} }^{p_l}_r,{\boldsymbol{\alpha}}^{q_l}_{t}\right)$ denote the discrete representations of ${a}_{r}^{i}\left({\boldsymbol{\alpha}}^{l}_{r}\right)$, ${a}_{t}^{j}\left({\boldsymbol{\alpha}}^{l}_{t}\right)$ and $H_{a}\left({\boldsymbol{\alpha}}^{l}_r,{\boldsymbol{\alpha}}^{l}_{t}\right)$ in (\ref{SISO_Continuous});
	${\bf a}_{r,i}^{l}\triangleq\left[{a}_{r}^{i}\left({\boldsymbol{\alpha}}_{r}^{1}\right), \cdots, {a}_{r}^{i}\left({\boldsymbol{\alpha}}_{r}^{d_r}\right)\right]^{H}\in \mathbb{C}^{d_r\times 1}$ represents the receive response of the $i$-th receive element in the $l$-th arrival cluster. Similarly, ${\bf a}_{t,j}^{l}\triangleq\left[{a}_{t}^{j}\left({\boldsymbol{\alpha}}_{t}^{1}\right), \cdots, {a}_{t}^{j}\left({\boldsymbol{\alpha}}_{t}^{d_t}\right)\right]^{H}\in \mathbb{C}^{d_t\times 1}$ denotes the transmit response of the $j$-th transmit element in the $l$-th departure cluster, ${\bf H}_{a}^{l}\in \mathbb{C}^{d_r^l\times d_t^l}$ represents the scattering response of the $l$-th cluster. ${\bf{q}}_{g}=[q_{g}^{1},\cdots,q_{g}^{N_g}]^H\in \mathbb{C}^{N_g\times 1},g\in\left\{r,t\right\}$ represents the adjustable phase responses. ${\bf A}_{g}^{l}=\left[{\bf a}^{l}_{g,1},\cdots,{\bf a}^{l}_{g,N_g}\right]\in \mathbb{C}^{d_g^l\times N_g}$ represents the receive or transmit response of the $l$-th cluster. Finally, ${\bf H}_{w}^{l}\in \mathbb{C}^{N_r\times N_t}$ denotes array-scattering response of $l$-th cluster, and array-scattering response across all $L$ clusters is ${\bf H}_{w}=\frac{1}{\sqrt{L}}\sum\limits_{l=1}^{L}{\bf H}_{w}^{l}\in \mathbb{C}^{N_r\times N_t}$.

	From (\ref{final_siso_channel}), we observe that the DARISA SISO spatial channel comprises five components: the adjustable phase responses of transmitter and receiver ${\bf{q}}_{t}$ and ${\bf{q}}_{r}$, the transmit response ${\bf A}_{t}^{l}$, the receive response ${\bf A}_{r}^{l}$, and the scattering response ${\bf H}_{a}^{l}$. The conventional channel models do not fully consider the influence of the transceiver antennas. Considering the controllability of DARISA, the DARISA can achieve not only different phase responses in various spatial directions but also impart controllable time-variant properties.
	
	\subsection{DARISA MIMO Spatial Channel Model}\label{MIMO-Channel-model}

	In this subsection, we generalize the DARISA SISO spatial channel model to a DARISA MIMO system, as illustrated in Fig. 3. The wireless scattering environment contains $\tilde{L}$ clusters, with $\tilde{d}_{t}^{\tilde l}$ departing and $\tilde{d}_{r}^{\tilde l}$ arriving rays in the $\tilde{l}$-th cluster ($\tilde{l}=1,\cdots,\tilde{L}$). The $M$ transmit-DARISAs and $N$ receive-DARISAs are tightly arranged to form transmit and receive arrays of sizes $\tilde{D}_{t,x}\times \tilde{D}_{t,y}=M\Delta_tN_{t,x}\times \Delta_tN_{t,y}$ and $\tilde{D}_{r,x}\times \tilde{D}_{r,y}=N\Delta_rN_{r,x}\times \Delta_rN_{r,y}$, respectively. Each transmit/receive antenna comprises $N_g,g\in\left\{r,t\right\}$ phase-adjustable metasurface elements, resulting in $N{N}_{r}$ and $M{N}_{t}$ elements in the transmit and receive arrays.
	
	Based on this configuration, the DARISA MIMO spatial channel $\widetilde{\bf{H}}\in{{\mathbb{C}}^{N\times M}}$ is expressed as
	\begin{align}
		\label{final_MIMO_channel}
		\widetilde{\bf{H}}&=\frac{1}{\sqrt{\tilde L}}\sum\limits_{\tilde l=1}^{\tilde L}{\widetilde{\bf{Q}}_{r}^H\underbrace{\widetilde{\bf A}_{r}^{\tilde l H}\widetilde{\bf H}_{a}^{\tilde l}\widetilde{\bf A}_{t}^{\tilde l}}_{\widetilde{\bf H}_{w}^{\tilde l}}\widetilde{\bf{Q}}_{t}}={\widetilde{\bf{Q}}_{r}^H\widetilde{\bf H}_{w}\widetilde{\bf{Q}}_{t}},
	\end{align}
	where $\widetilde{\bf A}_{r}^{\tilde{l}}\in \mathbb{C}^{{\tilde d}_r^{\tilde{l}}\times NN_r}$, $\widetilde{\bf A}_{t}^{\tilde{l}}\in \mathbb{C}^{{\tilde d}_t^{\tilde{l}}\times MN_t}$ and ${\bf H}_{w}^{\tilde l}\in \mathbb{C}^{NN_r\times MN_t}$ correspond to ${\bf A}_{r}^{l}$, ${\bf A}_{t}^{{l}}$ and ${\bf H}_{w}^{l}$ in (\ref{final_siso_channel}). Additionally, ${\widetilde{\bf{Q}}_{t}}\triangleq\textrm{blkdiag}\left(\widetilde{\bf{q}}_{t}^{1},\cdots,\widetilde{\bf{q}}_{t}^{M}\right)\in{{\mathbb{C}}^{MN_t\times M}}$ and  ${\widetilde{\bf{Q}}_{r}}\triangleq\textrm{blkdiag}\left(\widetilde{\bf{q}}_{r}^{1},\cdots,\widetilde{\bf{q}}_{r}^{N} \right)\in{{\mathbb{C}}^{NN_r\times N}}$ are the adjustable phase response matrices of the transmit and receive arrays, respectively. Here, $\widetilde{\bf{q}}_{t}^{m}=[q_{t}^{m,1},\cdots,q_{t}^{m,N_t}]^H\in \mathbb{C}^{N_t\times 1},m \in \mathcal{M}\triangleq \left\{1,\cdots,M\right\}$ and $\widetilde{\bf{q}}_{r}^{n}=[q_{r}^{n,1},\cdots,q_{r}^{n,N_r}]^H\in \mathbb{C}^{N_r\times 1},n\in \mathcal{N}\triangleq \left\{1,\cdots,N\right\}$ represent the adjustable phase responses of the $m$-th transmit antenna and the
	the $n$-th receive antenna, respectively. The phase responses $q_{t}^{m,j}=e^{\jmath\varphi_{t}^{m,j}},j\in \mathcal{J}$ and $q_{r}^{n,i}=e^{\jmath\varphi_{r}^{n,i}},i\in \mathcal{I}$ correspond to the $j$-th transmit element and the $i$-th receive element. Finally, $\widetilde{\bf{H}}_w\in \mathbb{C}^{NN_r\times MN_t}$ represents the array-scattering response of all $\tilde{L}$ clusters, determined by the arrangement of the elements and scattering environment, i.e.,
	\begin{align}
		\label{Scatter_MIMO_channel}
		\widetilde{\bf{H}}_w=\frac{1}{\sqrt{\tilde{L}}}\sum\limits_{\tilde{l}=1}^{\tilde{L}}\widetilde{\bf A}_{r}^{\tilde{l}H}\widetilde{\bf H}_{a}^{\tilde{l}}\widetilde{\bf A}_{t}^{\tilde{l}}=\frac{1}{\sqrt{\tilde L}}\sum\limits_{\tilde l=1}^{\tilde L}\widetilde{\bf H}_{w}^{\tilde l},
	\end{align}
	
	The received signal under the channel (\ref{final_MIMO_channel}) is expressed as
	\begin{align}
		{\bf y}=\widetilde{\bf{H}}{\bf x}+{\bf w},
	\end{align}
	where ${\bf x}=[x_1,\cdots,x_M]^T\in{{\mathbb{C}}^{M\times 1}}$ is the emitted signal, ${\bf w}=[w_{r,1},\cdots,w_{r,N}]^T\in{{\mathbb{C}}^{N\times 1}}$ is the Gaussian noise.
	
	\subsection{DARISA MIMO Spatial-Temporal Channel Model}\label{MIMO-Channel-model}
	The real-time controllability of the DARISA allows the agile phase response time of the metasurface to operate within the nanosecond level \cite{Sleasman}. This enables dynamic agile phase response adjustment of metasurface elements in the space-time domain within one symbol duration. Consequently, we extend the DARISA MIMO spatial channel (\ref{final_MIMO_channel}) to a composite spatial-temporal channel, where all parameters must account for the time variable $t$. The received signal is then re-expressed as
	\begin{align}
		\notag
		{\bf y}(t)&=\widetilde{\bf{H}}(t) {\bf x}(t)+ {\bf w}(t)=\widetilde{\bf{Q}}_{r}^H(t)\widetilde{\bf H}_{w}(t)\widetilde{\bf{Q}}_{t}(t){\bf x}(t)+ {\bf w}(t),
	\end{align}
	For further discussion, we establish the following assumptions:
	
	1. The $\tilde L$ cluster scattering channel remains constant during the symbol emission duration $T$, i.e., $\widetilde{\bf H}_{w}^{\tilde l}(t_k)=\widetilde{\bf H}_{w}^{\tilde l},\ {\bf x}(t_k)={\bf x},\ 0<t_k<T$. This assumption is valid under typical slow-fading and block-fading scenarios.
	
	2. The phase responses of the metasurface elements can be adjusted continuously within $(0,2\pi]$, i.e., $\left[\widetilde{\bf{Q}}_{r}\right]_{n,i}=e^{\jmath\varphi_{r}^{n,i}}$, $\left[\widetilde{\bf{Q}}_{t}\right]_{m,j}=e^{\jmath\varphi_{t}^{m,j}}$, $\varphi_{r}^{n,i}\in  (0,2\pi],\varphi_{t}^{m,j}\in  (0,2\pi]$.
	
	Based on these assumptions, we agilely adjust $K$ phase responses of metasurface elements during one symbol period. The resulting $K$-adjustments received signals are stacked into a $KN$-dimensional received signal vector as
	\begin{align}
		\notag
		\bar{\bf y} & =
		\left[
		\begin{array}{c}
			\widetilde{\bf{H}}(t_1){\bf x}(t_1) \\
			\vdots  \\
			\widetilde{\bf{H}}(t_K){\bf x}(t_K)
		\end{array}
		\right]  +
		\left[
		\begin{array}{c}
			{\bf w}(t_1) \\
			\vdots  \\
			{\bf w}(t_K)
		\end{array}
		\right] \\
		\notag
		&= \left[
		\begin{array}{c}
			\widetilde{\bf{Q}}_{r}^H(t_1)\widetilde{\bf H}_{w}\widetilde{\bf{Q}}_{t}(t_1) \\  
			\vdots  \\
			\widetilde{\bf{Q}}_{r}^H(t_K)\widetilde{\bf H}_{w}\widetilde{\bf{Q}}_{t}(t_K)
		\end{array}
		\right]{\bf x} +\overline{\bf w}\\
	\notag
		&= \overline{\bf{Q}}_{r}^H \overline{\bf H}_w\overline{\bf{Q}}_{t}{\bf x}+\overline{\bf w}\\
			\label{spatial-temporal-MIMO_channel}
			&=\overline{\bf H}_C{\bf x}+\overline{\bf w},
	\end{align}
	where $\overline{\bf{Q}}_{t}=\left[\widetilde{\bf{Q}}_{t}^T(t_1);\cdots; \widetilde{\bf{Q}}_{t}^T(t_K)\right]^T\in \mathbb{C}^{KMN_t\times M}$ and $\overline{\bf{Q}}_{r}=\textrm{blkdiag}\left(\widetilde{\bf{Q}}_{r}(t_1),\cdots, \widetilde{\bf{Q}}_{r}(t_K)\right)\in \mathbb{C}^{KNN_r\times KN}$ represent the spatial-temporal adjustable phase responses of the transmit and receive arrays. The spatial-temporal array-scattering response of all $L$ clusters is given by $\overline{\bf H}_w=\textrm{blkdiag}\left(\widetilde{\bf H}_{w}\cdots,\widetilde{\bf H}_{w}\right)\in \mathbb{C}^{KNN_r\times KMN_t}$, while
	$\overline{\bf H}_C \in \mathbb{C}^{KN\times M}$ denotes the composite spatial-temporal channel of all $\tilde L$ clusters. Additionally,
	$\overline{\bf w}\in{{\mathbb{C}}^{KN\times 1}}$ represents the spatial-temporal Gaussian white noise.
	
	From the above analysis, it is evident that the $KN\times M$-dimensional equivalent MIMO channel matrix (\ref{spatial-temporal-MIMO_channel}) is reconstructed by agilely adjusting the phase responses without increasing the signal bandwidth. The MIMO system fully leverages the spatial resources of the channel by employing multiple antennas at the transceiver, thereby providing multiplexing gain (DoF). The DoF is contingent upon the performance of the composite spatial-temporal channel $\overline{\bf H}_C$ in (\ref{spatial-temporal-MIMO_channel}), which will be analyzed in detail in the next section.

	\section{ DoF analysis of DARISA MIMO }\label{DoF-analysis}
	In MIMO systems, the DoF refers to the number of independent communication data streams that the transmission system can support. Also known as multiplexing gain, the DoF characterizes the rate at which communication capacity increases with the signal-to-noise ratio (SNR) at a high SNR region. From (\ref{spatial-temporal-MIMO_channel}), it is evident that the DARISA MIMO system constructs an equivalent spatial-temporal extended channel through agilely adjusting phase responses of metasurface elements. According to the  MIMO communication theory, the rank of the composite spatial-temporal channel $\overline{\bf H}_C$ in (\ref{spatial-temporal-MIMO_channel}) represents the DoF. Specifically,
	$\overline{\bf H}_C={\overline{\bf{Q}}}_{r}^H\overline{\bf H}_w{\overline{\bf{Q}}}_{t}$ encompasses the spatial-temporal adjustable phase responses of transmit and receive DARISA arrays $\overline{\bf{Q}}_{r}$ and $\overline{\bf{Q}}_{t}$, along with the total spatial-temporal array-scattering response of all $\tilde L$ clusters $\overline{\bf H}_w$. Thus, the rank of $\overline{\bf H}_C$ is influenced by the number of transmit-DARISA and receive-DARISA $M$ and $N$, the number of each transmit-DARISA and receive-DARISA elements $N_t$ and $N_r$, agility frequentness $K$, and the scattering response of wireless environment, including the number of clusters and the angular spread of those clusters.
	
	Since each metasurface element of the DARISA can independently adjust phase responses, appropriately designing the responses can render the columns of $\overline{\bf{Q}}_{r}$ and $\overline{\bf{Q}}_{t}$ full rank, making the rank of $\overline{\bf H}_C$ dependent on $\overline{\bf H}_w$ in (\ref{spatial-temporal-MIMO_channel}). $\overline{\bf H}_w=\textrm{diag}\left\{\widetilde{\bf H}_w,\cdots,\widetilde{\bf H}_w \right\}$ implies that $\widetilde{\bf H}_w$ in (\ref{Scatter_MIMO_channel}) determines the rank of $\overline{\bf H}_w$. Since $\widetilde{\bf H}_w$ is determined by the arrangement of elements and the scattering environment, which cannot be artificially controlled, we first analyze the rank of $\widetilde{\bf H}_{w}$.
	
	According to the literature {\cite{Poon}}, the rank of $\widetilde{\bf H}_{w}$ is jointly determined by the angular spread of all clusters and the sizes of the transceiver arrays. We prove the following lemma.
	\begin{lemma}\label{lemma1}
		$\textrm{rank}\left(\widetilde{\bf H}_{w}\right)=\min\left(\tilde{d}_{r},\tilde{d}_{t}\right),$
where $\tilde{d}_g\approx{c}_{g,1}{c}_{g,2}\pi\tilde{D}_{g,x} \tilde{D}_{g,y}, g\in\left\{r,t\right\}$, $\tilde{D}_{g,x} \tilde{D}_{g,y}$ is the sizes of transmit/receive DARISA array, $0\leq {c}_{g,1}\leq 1$ and $0\leq {c}_{g,2}\leq 1$ represent the semi-major and semi-minor axes, respectively, of the tightest ellipse encompassing all space-domain clusters projected onto the wavenumber domain.
	\end{lemma}
	\begin{proof}
		The proof is provided in Appendix \ref{proof1}.
	\end{proof}
	
	From Lemma \ref{lemma1}, we conclude that the rank of $\widetilde{\bf H}_{w}$ is equal to $\min\left(\tilde{d}_{r},\tilde{d}_{t}\right)$. To further analyze the rank of the composite spatial-temporal channel $\overline{\bf H}_C$ in (\ref{spatial-temporal-MIMO_channel}), we present the following theorem.
	\begin{thm}\label{theorem1}
		$\textrm{rank}({\overline{\bf H}}_C)= \min\left(KN,M,\tilde{d}_t,\tilde{d}_r\right)$.
	\end{thm}
	\begin{proof}
		The detail proof can be found Appendix \ref{proof-of-theorem1}.
	\end{proof}
	Theorem \ref{theorem1} represents the DoF of the DARISA MIMO system under the DAARP strategy. According to this theorem, we have the following remarks.
	
	\textbf{Remark 1: } The DoF of the DARISA MIMO system under DAAPR strategy is $\min\left(KN,M,\tilde{d}_t,\tilde{d}_r\right)$. This value is determined by the number of transmit-DARISAs and receive-DARISAs $N,M$, agility frequentness $K$, the array sizes of DARISA-transceiver $ \tilde{D}_{g,x}\tilde{D}_{g,y},g\in\left\{r,t\right\}$, and scattering environment ${c}_{g,1},{c}_{g,2}$. This implies that $\min\left(\tilde{d}_t,\tilde{d}_r\right)$ serves as the upper bound on the DoF influenced by the array size and scattering environment, which cannot be exceeded.
	
		\textbf{Remark 2: } When the number of transmit-DARISAs exceeds that of receive-DARISAs ($M>N$), the DoF is limited by $N$. Agilely adjusting $K$ phase responses of metasurface elements can improve the row rank of $\overline{\bf H}_C$, extending the DoF to $\min\left(KN,M,\tilde{d}_t,\tilde{d}_r\right)$. In this case, agilely adjusting the phase responses in the time domain can improve the DoF without increasing the number of receive-DARISAs and signal bandwidths. Conversely, if the number of transmit antennas is fewer than the number of receive antennas ($M<N$), the DoF is limited by $M$, resulting in $\textrm{rank}\left(\overline{\bf H}_C\right)= \min\left(M,\tilde{d}_t,\tilde{d}_r\right)$; in this case, agilely adjusting phase responses does not increase the DoF.
	
		\textbf{Remark 3: } According to the composite spatial-temporal channel (\ref{spatial-temporal-MIMO_channel}), we find that agilely adjusting phase responses of metasurface elements at both transmitter and receiver or solely at the receiver can increase the row rank $KN$ of $\overline{\bf H}_C$ but not the column rank $M$. It implies that the DoFs are the same by agilely adjusting phase responses at both transmitter and receiver or only at the receiver. For agilely adjusting the phase responses at the transmitter, the duration of each symbol is only $1/K$ of the original duration, which corresponds to an expansion of occupied bandwidth.

	\section{Effective DoF Analysis and Optimization}
		The theoretical analysis of the DoF for DARISA MIMO systems under the DAAPR strategy is presented in Section III. Theorem 1 establishes that the DoF is determined by the number of transceiver DARISAs, agility frequentness, array sizes, and scattering environment. According to the fundamental MIMO theory, in the high-SNR regime, the DoF corresponds to the number of independent subchannels for parallel data streams, dictating the slope of channel capacity growth. However, channel capacity depends on both the DoF and the singular value distribution of the channel matrix. The magnitude of the singular value indicates the quality of that parallel subchannel for information transmission.
		To enable a more precise characterization, EDoF is introduced to account for both the quantity and quality of subchannels. While DoF analysis provides the theoretical foundation for dynamic agile adjustment design, EDoF optimization refines capacity enhancements by addressing subchannel efficiency. These two metrics form a complementary analytical framework. Consequently, this section focuses on maximizing EDoF by optimizing the agile phase responses of metasurface elements, directly linking EDoF improvement to channel capacity.
	
	\subsection{ Evolutionary Motivation from DoF to EDoF}\label{EDOF_analysis}
	We rewrite (\ref{spatial-temporal-MIMO_channel}) more briefly as
	\begin{align}
		\bar{\bf y}= \overline{\bf H}_C{\bf x}+\overline{\bf w},
	\end{align}
	where $\bar{\bf y}\in \mathbb{C}^{KN\times 1},\overline{\bf H}_C\in \mathbb{C}^{KN\times M}$ and $\overline{\bf w}\in \mathbb{C}^{KN\times 1}$ are expressed in (\ref{spatial-temporal-MIMO_channel}). Each element of $\overline{\bf w}$ obeys a complex Gaussian distribution with zero mean and $\delta^2$ variance.
	
Based on the work \cite{Tse}, channel capacity is jointly determined by the rank of $\overline{\mathbf{H}}_C$ and the distribution of its singular values. At low SNR, optimal capacity is achieved by allocating all power to the sub-channel with the largest singular value, which is independent of DoF.	At high SNR, allocating equal power across all sub-channels becomes asymptotically optimal, leading to an approximation of capacity as
	\begin{align}
		\notag
		C&\approx \textrm{rank}\left(\overline{\mathbf{H}}_C\right)\log(\textnormal{SNR})+\sum_{a=1}^{\textrm{rank}\left(\overline{\mathbf{H}}_C\right)}\log\left( \frac{\lambda_a^2}{\textrm{rank}\left(\overline{\mathbf{H}}_C\right)}\right),
	\end{align}
	where ${\textnormal{SNR}}=\frac{P}{KM\delta^2}$. The channel capacity is influenced by both the rank of $\overline{\mathbf{H}}_C$ and the distribution of its singular values $\lambda_a^2$. Therefore, we primarily focuses on capacity at high SNR. The DoF represented by $\textrm{rank}\left(\overline{\mathbf{H}}_C\right)$ indicates the number of independently transmitted information streams. The singular values $\lambda_a^2$ reflect the distribution of subchannel through the condition number of MIMO channel, defined as $\frac{\max_a \lambda_a^2}{\min_a \lambda_a^2}$. A condition number close to 1 indicates a well-conditioned channel, which corresponds to maximal channel capacity. Existing studies on HMIMOS \cite{Yuan1}-\cite{Yuan3} that focus on channel DoF perspective is flawed because they neglect the influence of singular value distribution.
	
	To address this issue, the EDoF is utilized, which considers both channel DoF and the distribution of singular values \cite{Muharemovic}-\cite{Verdu}. Consequently, channel capacity can be approximated as
	\begin{align}
		\label{EDoF-capacity}
		C\approx \Psi\left(\overline{\mathbf{H}}_C\right)\log_2\left(1+\frac{\textnormal{SNR}}{\Psi\left(\overline{\mathbf{H}}_C\right)}\right),
	\end{align}
	where the EDoF of $\overline{\mathbf{H}}_C$ is denoted as
	\begin{align}
		\label{Original-EDoF}
		\Psi\left(\overline{\mathbf{H}}_C\right)\triangleq\left(\frac{\textrm{Tr}\left(\overline{\mathbf{H}}_C\overline{\mathbf{H}}_C^H\right)}{\big{|}\big{|}\overline{\mathbf{H}}_C\overline{\mathbf{H}}_C^H\big{|}\big{|}_F}
		\right)^2=\frac{\left(\sum_{a=1}^{\textrm{rank}\left(\overline{\mathbf{H}}_C\right)}\lambda_a^2\right)^2}{\sum_{a=1}^{\textrm{rank}\left(\overline{\mathbf{H}}_C\right)}\lambda_a^4}.
	\end{align}
	We establish Lemma \ref{lemma2} to prove that channel capacity monotonically increases with EDoF. Therefore, the EDoF is critical.
	
	\begin{lemma}\label{lemma2}
	The channel capacity monotonically increases with EDoF.
	\end{lemma}
	\begin{proof}
		The proof is provided in Appendix \ref{proof_lemma2}.
	\end{proof}
	
	\subsection{Analysis of Effective DoF}\label{EDOF_analysis}
	In this subsection, we will analyze EDoF in detail.
	
	Based on (\ref{Original-EDoF}), the range of EDoF is given by $\Psi\left(\overline{\mathbf{H}}_C\right) \in \left[1,\textrm{rank}\left(\overline{\mathbf{H}}_C\right)\right]$. The EDoF is minimal when $\textrm{rank}\left(\overline{\mathbf{H}}_C\right)=1$, and is maximal when $\overline{\mathbf{H}}_C$ achieves full rank with all eigenvalues equal. This implies that the channel DoFs based on Theorem 1 is the theoretical upper bound of EDoF.
	
	The EDoF is dependent on $\overline{\mathbf{H}}_C= {\overline{\bf{Q}}}_{r}^H\overline{\bf H}_w{\overline{\bf{Q}}}_{t}$ in (\ref{spatial-temporal-MIMO_channel}). To maximize the EDoF, by giving $\overline{\mathbf{H}}_w$, $\overline{\bf{Q}}_{r}$ and $\overline{\bf{Q}}_{t}$ should be optimized. To further analyze the impact of optimization results of $\overline{\bf{Q}}_{r}$ and $\overline{\bf{Q}}_{t}$ on the EDoF, we present Lemma \ref{lemma3}.
	\begin{lemma}\label{lemma3}
		When the phase responses of the metasurface elements can be adjusted continuously within $(0,2\pi]$, optimizing the phases of both $\overline{\bf{Q}}_{t}$ and $\overline{\bf{Q}}_{r}$ is equivalent to optimizing the phase of $\overline{\bf{Q}}_{r}$ only.
	\end{lemma}
	\begin{proof}
		The proof is provided in Appendix \ref{proof_lemma3}.
	\end{proof}
	With Lemma \ref{lemma3}, only optimizing $\overline{\bf{Q}}_{r}$ is equivalent to optimizing both $\overline{\bf{Q}}_{r}$ and $\overline{\bf{Q}}_{t}$, so we optimize $\overline{\bf{Q}}_{r}$ to maximize the EDoF in the next subsection.
	
	\textbf{Remark 4:}	According to (\ref{Original-EDoF}), if all singular values $\lambda_a^2$ of $\overline{\mathbf{H}}_C$ are equal, the EDoF is equal to the DoF. Based on Lemma \ref{lemma3}, only optimizing $\overline{\bf{Q}}_{r}$ is equivalent to optimizing both $\overline{\bf{Q}}_{r}$ and $\overline{\bf{Q}}_{t}$. Therefore, if $\overline{\bf{Q}}_{r}$ is optimized to satisfy that all $\lambda_a^2$ are equal, the maximum EDoF is given by $\textrm{rank}\left(\overline{\mathbf{H}}_C\right)= \min\left(KN,M,\tilde{d}_t,\tilde{d}_r\right)$. Since agile dynamic adjustments of phase response guarantees $KN\geq M$, the EDoF further equals $\min\left(M,\tilde{d}_t,\tilde{d}_r\right)$. If $M<\left(\tilde{d}_t,\tilde{d}_t\right)$, the EDoF equals $M$. Conversely, when $M\geq\left(\tilde{d}_t,\tilde{d}_r\right)$, the EDoF is determined by $\left(\tilde{d}_t,\tilde{d}_r\right)$, which indicates that the maximum EDoF increases with array sizes and angular spreads of the wireless channel.

	\subsection{Optimization of Effective DoF}
	In this subsection, we optimize $\overline{\bf{Q}}_{r}$ to maximize the EDoF with continuous and discrete phases, respectively.
	
	\subsubsection{Optimizing EDoF with Continuous Phase}
	First, considering the phase of each element can be continuously designed among $\left[0,2\pi\right)$, we optimize $\overline{\bf{Q}}_{r}$ to maximize the EDoF by fixing the number of elements of each receive-DARISA $N_r$ and agility frequentness $K$. The optimization problem is expressed as
	\begin{subequations}\label{edof}
		\begin{align}
			\max_{\varphi_{r}^{n,i}(t_k)}\ & \frac{\textrm{Tr}\left(\overline{\mathbf{H}}_C\overline{\mathbf{H}}_C^H\right)}{\big{|}\big{|}\overline{\mathbf{H}}_C\overline{\mathbf{H}}_C^H\big{|}\big{|}_F},\\
			\ s.t. \quad &\quad  \varphi_{r}^{n,i}(t_k)\in (0,2\pi], i\in \mathcal{I}, n\in \mathcal{N}, k\in \mathcal{K},
		\end{align}
	\end{subequations}
	where $\varphi_{r}^{n,i}(t_k)$ is the phase response of $i$-th element of $n$-th receive-DARISA at $k$-th agility frequentness, $\mathcal{I}, \mathcal{U}$ are denoted in (\ref{final_MIMO_channel}), $\mathcal{K}\triangleq \left\{1,\cdots,K\right\}$.
	The problem is non-convex due to the objective function. To solve this problem, we denote
	\begin{align}
		\notag
		&{\bf C}\triangleq {\overline{\mathbf{H}}}_{w}{\overline{\bf{Q}}}_{t}{\overline{\bf{Q}}}_{t}^H{\overline{\mathbf{H}}}_{w}^H\in \mathbb{C}^{{KNN_r\times KNN_r}},\\
		\notag
		& {\bf E}\triangleq {\overline{\bf{Q}}}_{r}{\overline{\bf{Q}}}_{r}^H=\textrm{blkdiag}\left(\widetilde{\bf{q}}_{r}^{1}(t_1)\widetilde{\bf{q}}_{r}^{1H}(t_1),\cdots, \widetilde{\bf{q}}_{r}^{N}(t_1)\widetilde{\bf{q}}_{r}^{NH}(t_1),\right.\\
		\notag
		&\left.\cdots, \widetilde{\bf{q}}_{r}^{1}(t_K)\widetilde{\bf{q}}_{r}^{1H}(t_K),\cdots, \widetilde{\bf{q}}_{r}^{N}(t_K)\widetilde{\bf{q}}_{r}^{NH}(t_K)\right)\in \mathbb{C}^{{KNN_r\times KNN_r}},
	\end{align}
	where ${\bf E}$ is a block diagonal matrix with $\textrm{rank}\left({\bf E}\right)=KN$, the rank of each block diagonal matrix of $\bf E$ equal $1$ and all diagonal elements equal $1$. Since
	$\textrm{Tr}\left(\overline{\mathbf{H}}_C\overline{\mathbf{H}}_C^H\right)=\textrm{Tr}\left({\overline{\bf{Q}}}_{r}^H{\overline{\mathbf{H}}}_{w}{\overline{\bf{Q}}}_{t}{\overline{\bf{Q}}}_{t}^H
	{\overline{\mathbf{H}}}_{w}^H{\overline{\bf{Q}}}_{r}\right)=\textrm{Tr}\left({\bf CE}\right)$, problem (\ref{edof1}) is equivalently transformed as
	\begin{subequations}\label{edof1}
		\begin{align}
			\max_{{\bf E}} & \quad \frac{\textrm{Tr}\left({\bf C}{\bf E}\right)}{{\big{|}\big{|}{\bf C}{\bf E}\big{|}\big{|}_F}}, \label{edof1a}\\
			s.t. &\quad {\bf E}_{z,z}=1,{\bf E}_{c,d}=0,\textrm{rank}\left({\bf E}\right)=KN,\label{edof1b}\\
			& \quad z\in\mathcal{X},\ c\in \mathcal{X}\backslash\mathcal{Y},\ d \in \mathcal{X}\backslash\mathcal{Y},c\neq d,\label{edof1c}
		\end{align}
	\end{subequations}
	where $\mathcal{X}$ denotes the set $\left\{1,2,\cdots,KNN_r\right\}$, $\mathcal{Y}$ denotes the set $\left\{\left\{kNN_r+1,kNN_r+2,\cdots,(k+1)NN_r\right\}|\forall k\in \mathcal{K}\right\}$, $\mathcal{X}\backslash\mathcal{Y}$ represents the difference set between set $\mathcal{X}$ and set $\mathcal{Y}$. This problem still is non-convex due to the objective function and fixed-rank constraint $\textrm{rank}\left({\bf E}\right)=KN$.
	
	To address these issues, we introduce an auxiliary parameter $\zeta$ by the Dinkelbach transformation method \cite{Shen2}, \cite{W}. Then, we develop a semidefinite relaxation (SDR) method to ignore the fixed rank constraint $\textrm{rank}\left({\bf E}\right)=KN$, thereby transforming problem (\ref{edof1}) as
	\begin{subequations}\label{edof2}
		\begin{align}
			\max_{{\bf E}}\ & \textrm{Tr}\left({\bf C}{\bf E}\right)-\zeta{\big{|}\big{|}{\bf C}{\bf E}\big{|}\big{|}_F},\\
			\ s.t. \ & (\ref{edof1b}),\ (\ref{edof1c}).
		\end{align}
	\end{subequations}
	Given the parameter $\zeta$, this problem is convex. Next, we analyze the impact of $\zeta$ on the problem (\ref{edof1}).
	
	The optimal solution of (\ref{edof2}) varies for different values of $\zeta$. If there exists a unique $\zeta^{opt}$ such that the optimal objective function of (\ref{edof2}) equals $0$, $\varphi_{r}^{{n,i}^{opt}}(t_k)$, derived from the optimal solution ${\bf E}^{opt}$ of problem (\ref{edof2}), serves as the optimal solution for original problem (\ref{edof1}). This holds true because when
	\begin{align}
		\notag
		&\max_{{\bf E}}\textrm{Tr}\left({\bf C}{\bf E}\right)-\zeta^{opt}{\big{|}\big{|}{\bf C}{\bf E}\big{|}\big{|}_F}\\
		\notag
		&=\textrm{Tr}\left({\bf C}{\bf E}^{opt}\right)-\zeta^{opt}{\big{|}\big{|}{\bf C}{\bf E}^{opt}\big{|}\big{|}_F}=0,
	\end{align}
	$\zeta^{opt}$ must satisfy
	\begin{align}
		\notag
		\zeta^{opt}=\frac{\textrm{Tr}\left({\bf C}{\bf E}^{opt}\right)}{\big{|}\big{|}{\bf C}{\bf E}^{opt}\big{|}\big{|}_F}=\max_{\varphi_{r}^{n,i}(t_k)} \frac{\textrm{Tr}\left(\overline{\mathbf{H}}_C\overline{\mathbf{H}}_C^H\right)}{\big{|}\big{|}\overline{\mathbf{H}}_C\overline{\mathbf{H}}_C^H\big{|}\big{|}_F}.
	\end{align}
	Therefore, the range of $\zeta^{opt}$ is $\left[1,\sqrt{\overline{\mathbf{H}}_C}\right]$. From the above analysis,  the critical task is to find a $\zeta^{opt}\in \left[1,\sqrt{\overline{\mathbf{H}}_C}\right]$ such that the objective of (\ref{edof2}) equals $0$ to demonstrate that $\zeta^{opt}$ is unique.
	
	We establish Lemma \ref{lemma4} that the optimal objective function of problem (\ref{edof2}) strictly decreases with respect to $\zeta$. Thus, if there exists a $\zeta^{opt}$ for which the optimal objective function of (\ref{edof2}) equals 0, $\zeta^{opt}$ must be unique.
	\begin{lemma}[]\label{lemma4}
		$F(\zeta)=\max_{{\bf E}}\left\{\textrm{Tr}\left({\bf C}{\bf E}\right)-\zeta{\big{|}\big{|}{\bf C}{\bf E}\big{|}\big{|}_F}\right\}$ strictly monotonically decreases with respect to $\zeta$, when $\zeta^{'}<\zeta^{''}$ (any $\zeta\geq 0$), $F(\zeta^{''})<F(\zeta^{'})$.
	\end{lemma}
	\begin{proof}
		Assuming that ${\bf E}^{''}$ is the optimal solution of $F(\zeta^{''})$, since $\zeta^{'}<\zeta^{''}$ and any $\zeta\geq 0$, we obtain
		\begin{align}
			\notag
			F(\zeta^{''})&=\max_{\bf E}\left\{\textrm{Tr}\left({\bf C}{\bf E}\right)-\zeta^{''}{\big{|}\big{|}{\bf C}{\bf E}\big{|}\big{|}_F}\right\}\\
			\notag
			&=\textrm{Tr}\left({\bf C}{\bf E}^{''}\right)-\zeta^{''}{\big{|}\big{|}{\bf C}{\bf E}^{''}\big{|}\big{|}_F}\\
			\notag
			&<\textrm{Tr}\left({\bf C}{\bf E}^{''}\right)-\zeta^{'}{\big{|}\big{|}{\bf C}{\bf E}^{''}\big{|}\big{|}_F}\\
			\notag
			&\leq \max_{\bf E}\left\{\textrm{Tr}\left({\bf C}{\bf E}\right)-\zeta^{'}{\big{|}\big{|}{\bf C}{\bf E}\big{|}\big{|}_F}\right\}\\
			\notag
			&=F(\zeta^{'}).
		\end{align}
		The proof is completed.
	\end{proof}
	
	Next, we prove that there exists a unique $\zeta^{opt}$ such that the optimal objective function of (\ref{edof2}) equals $0$.

	\begin{thm}[]\label{thm2}
		For $\zeta\in \left[1,\sqrt{\textrm{rank}\left(\overline{\mathbf{H}}_C\right)}\right]$, where $\textrm{rank}({\overline{\bf H}}_C)= \min\left(KN,M,\tilde{d}_t,\tilde{d}_r\right)$, there is precisely one $\zeta^{opt}$ for which the objective function of problem (\ref{edof2}) equals 0.
	\end{thm}
	\begin{proof}
		The proof is provided in Appendix \ref{proof of thm2}.
	\end{proof}
	
	From Theorem \ref{thm2}, we conclude that when $\zeta\in \left[1,\sqrt{\textrm{rank}\left(\overline{\mathbf{H}}_C\right)}\right]$, there is a unique $\zeta^{opt}$ that makes the optimal objective function of problem (\ref{edof2}) equal to 0. This zero root corresponds to the optimal objective function of the original problem (\ref{edof1}), which can be determined using the bisection method.
	
	After determining the optimal $\zeta^{opt}$ corresponding to ${\bf E}^{opt}$ from problem (\ref{edof2}), we recover the fixed rank constraint $\textrm{rank}\left({\bf E}\right)=KN$ to obtain the phase response of each element $\varphi_{r}^{{n,i}^{opt}}(t_k)$. Since ${\bf E}$ is composed of $KN$ block diagonal matrices, with each block being a rank-1 matrix,
	\begin{align}
		\notag
		{\bf E}=\textrm{blkdiag}\left(\widetilde{\bf{q}}_{r}^{1}(t_1)\widetilde{\bf{q}}_{r}^{1H}(t_1),\cdots, \widetilde{\bf{q}}_{r}^{N}(t_K)\widetilde{\bf{q}}_{r}^{NH}(t_K)\right),
	\end{align}
	by denoting ${\bf E}_{n,k}\triangleq\widetilde{\bf{q}}_{r}^{n}(t_k)\widetilde{\bf{q}}_{r}^{nH}(t_k),n\in \mathcal{N}, k\in \mathcal{K}$, we use gaussian random method to recover the rank-1 constraint of each block diagonal matrix ${\bf E}^{opt}_{n,k}$ of ${\bf {E}}^{opt}$ \cite{So}, i.e., ${{\bf E}}^{opt}_{n,k}=\boldsymbol{\vartheta}_{n,k}^{opt}\boldsymbol{\vartheta}_{n,k}^{optH}$, where $\boldsymbol{\vartheta}_{n,k}^{opt}$ is the optimal rank one solution recovered by Gaussian-random method that maximizes the objective function of problem (\ref{edof1}).
	
	Finally, we obtain the optimal phase of each element, i.e.,
	\begin{align}
		\label{optimal-phase}
		\varphi_{r}^{{n,i}^{opt}}(t_k)=\arg\left(\boldsymbol{\vartheta}_{n,k}^{opt}(i)\right), i\in \mathcal{I}.
	\end{align}
	
	The total algorithm to solve the problem (\ref{edof}) is summarized in Algorithm 1. The computational complexity of the SDR algorithm in each iteration is $\mathcal{O}\left\{(KNN_r)^3\right\}$.
	
	\begin{algorithm}[!t]
		\caption{ 	 Proposed algorithm to solve problem (\ref{edof}) }
		\begin{algorithmic}
			\State 1:\textbf{Initialization:} Set convergence precision $\epsilon=10^{-3}$. Set iteration index $\chi=0$; Initialize the upper bound and the lower bound of auxiliary parameter  $\zeta_u^{(0)}=\sqrt{\textrm{rank}\left(\overline{\mathbf{H}}_C\right)}$ and $\zeta_l^{(0)}=1$, respectively.
			\Repeat
			\State 2: Set $\zeta_{mid}^{(\chi)} =\frac{1}{2} \left(\zeta_u^{(\chi)}+\zeta_l^{(\chi)}\right)$;
			\State 3:With given $\zeta^{(\chi)}= \zeta_{mid}^{(\chi)}$, solve (\ref{edof2}) to obtian ${\bf E}^{(\chi)}$;
			\State 4: Check the positivity of $f\left({\bf E}^{(\chi)}\right)=\textrm{Tr}\left({\bf C}{\bf E}^{(\chi)}\right)-$\\
				\qquad$\zeta^{(\chi)}{\big{|}\big{|}{\bf C}{\bf E}^{(\chi)}\big{|}\big{|}_F}$;
			\State 5: If $f\left({\bf E}^{(\chi)}\right)\geq 0$, update $\zeta_l^{(\chi+1)} = \zeta_{mid}^{(\chi)}$;
			\State 6: If $f\left({\bf E}^{(\chi)}\right)< 0$, update $\zeta_u^{(\chi+1)} = \zeta_{mid}^{(\chi)}$;
			\State 7: Set $\chi\leftarrow \chi+1$;
			\Until{$\zeta_u^{(\chi)}-\zeta_l^{(\chi)}\leq \epsilon$;}
			\State 8: Set ${\bf E}^{opt}={\bf E}^{(\chi)}$, recover $\textrm{rank}\left({\bf E}^{opt}\right)=KN$;
			\State 9: Obtain phase response of each element by (\ref{optimal-phase});
			\State 10: Compute the optimized EDoF $\frac{\textrm{Tr}\left(\overline{\mathbf{H}}_C\overline{\mathbf{H}}_C^H\right)}{\big{|}\big{|}\overline{\mathbf{H}}_C\overline{\mathbf{H}}_C^H\big{|}\big{|}_F}$ of (\ref{edof}).
		\end{algorithmic}
	\end{algorithm}
	
	\subsubsection{Optimizing EDoF with Discrete Phase}
	
	When the phase response of metasurface elements is discretely phased, by setting the quantization bits of each element as $b$, $\mathcal{Q}\triangleq\left\{0,\frac{\pi}{2^{b-1}},\cdots, \frac{\pi(2^b-1)}{2^{b-1}} \right\}$ is the discrete set of adjustable phases of each element, the problem of maximizing channel capacity is expressed as
	\begin{align}
		\notag
		\max_{\varphi_{r}^{n,i'}(t_k)}&\quad \frac{\textrm{Tr}\left(\overline{\mathbf{H}}_C\overline{\mathbf{H}}_C^H\right)}{\big{|}\big{|}\overline{\mathbf{H}}_C\overline{\mathbf{H}}_C^H\big{|}\big{|}_F},\\
		\label{edof_dis}
		s.t. \quad &\quad  \varphi_{r}^{n,i'}(t_k)\in \mathcal{Q}, i\in \mathcal{I}, n\in \mathcal{N}, \forall k\in \mathcal{K},
	\end{align}
	To solve this non-convex problem, we first relax the discrete phase constraint $ \varphi_{r}^{n,i'}(t_k)\in \mathcal{Q}$ and convert it into a continuous phase constraint $ \varphi_{r}^{n,i}(t_k)\in (0,2\pi]$. This transformation aligns problem (\ref{edof_dis}) with problem (\ref{edof}), allowing us to solve for the optimal continuous phase of each element $\varphi_{r}^{{n,i}^{opt}}(t_k)\in (0,2\pi]$ by SDR method. Subsequently, we apply the principle of minimum mean square error to recover the discrete phase, yielding the optimal discrete phase for each element as
	\begin{align}
		 \varphi_{r}^{{n,i}^{opt'}}(t_k)=\arg\min\limits_{ \varphi_{r}^{n,i'}(t_k)\in \mathcal{Q}}| \varphi_{r}^{{n,i}^{opt'}}(t_k)- \varphi_{r}^{n,i'}(t_k)|.
  \end{align}
	
	\section{Simulation Results}
	
	This section presents simulation results to verify the performance of the DARISA MIMO system under various scenarios. We assume that each cluster follows the Cluster Delay Line (CDL) channel model in 3GPP TR 38.900.  In addition, all simulation results ignore the effect of mutual coupling characteristics on the performance, which needs to be investigated in the future.
	
		\begin{figure}
		\centering
		\includegraphics[width=3.8in]{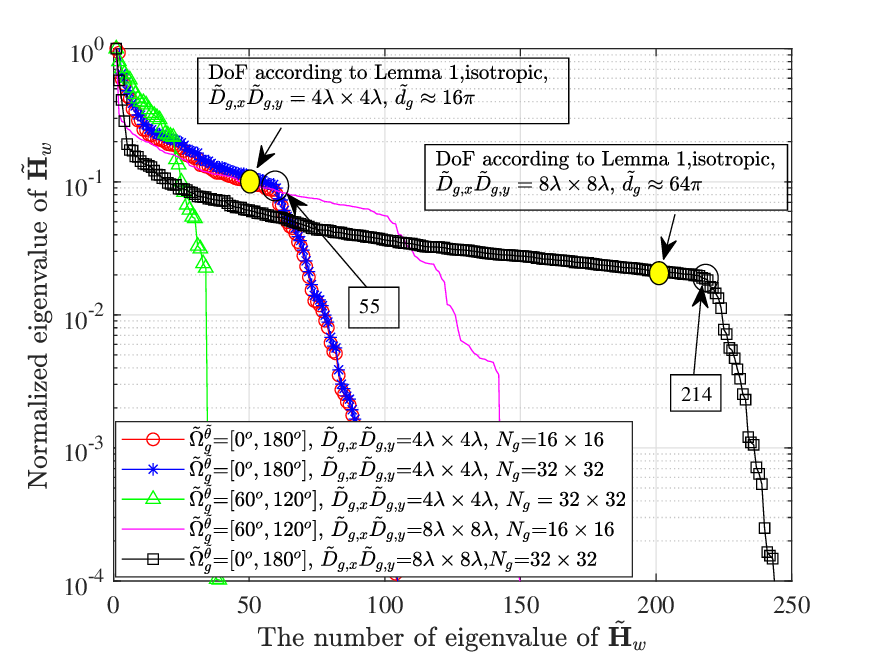}
		\caption{ The DoF of the MIMO array-scattering channel $\widetilde{\bf H}_{w}$ in (\ref{Scatter_MIMO_channel}) and Lemma \ref{lemma1} for different DARISA MIMO sizes and azimuth angular spreads.}
	\end{figure}
	
	\begin{figure}
		\centering
		\subfigure[ The eigenvalue of composite spatial-temporal channel $\overline{\bf H}_C$ in  (\ref{spatial-temporal-MIMO_channel}) with random-phase (RP) scheme and our proposed optimized-phase (OP) scheme.]{\includegraphics[width=3.8in]{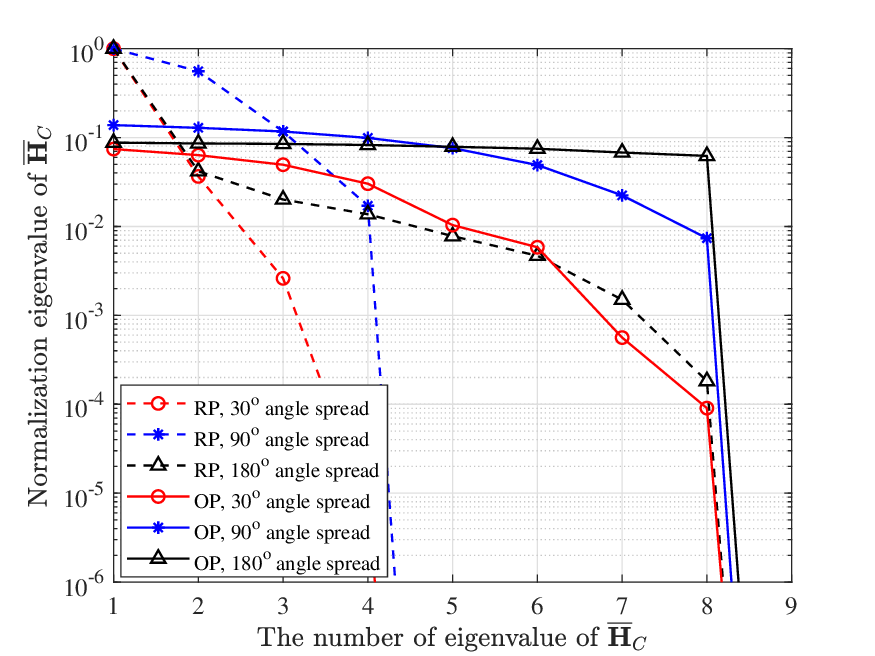}}
		\subfigure[ The capacity versus SNR with different phase schemes and angular spreads.]{\includegraphics[width=3.8in]{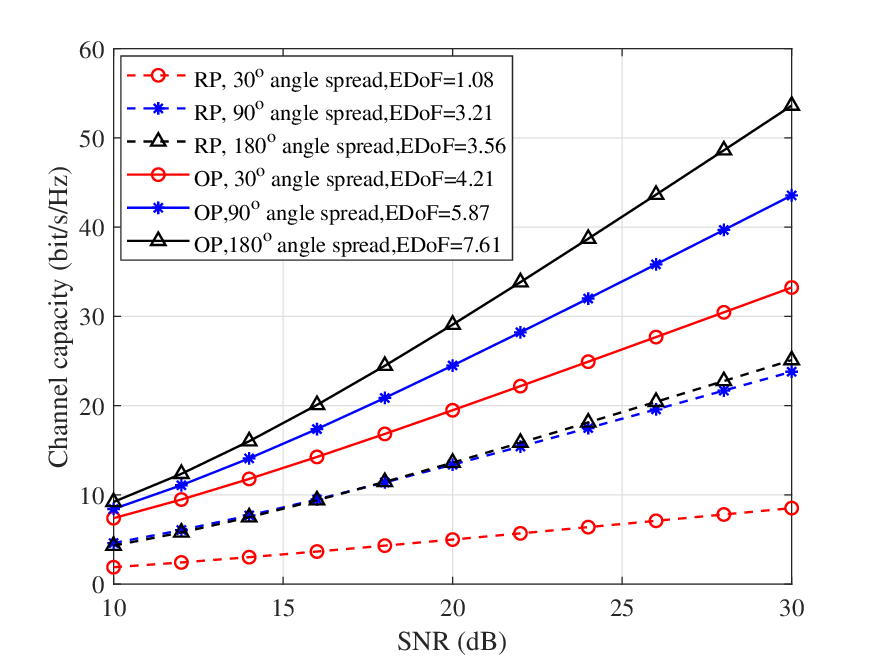}}
		\caption{ The eigenvalue and capacity with different phase shifts and angular spreads.}
	\end{figure}

	Fig. 4 illustrates the DoF of array-scattering channel $\widetilde{\bf H}_{w}$ (defined in (\ref{Scatter_MIMO_channel})) for different DARISA MIMO array sizes and azimuth angular spreads. The yellow circular marker denotes theoretical DoF values derived from Lemma 1 under isotropic scattering conditions. The remaining curves correspond to 3GPP TR 38.900 CDL channel model simulations \cite{TR}. First, we set the cluster parameters including the number and angles of clusters based on the 3GPP CDL model. Then,  the coefficients of the scattering response matrix $\widetilde{\bf{H}}_w$ in (\ref{Scatter_MIMO_channel}) are generated by utilizing these parameters. Finally, the singular value decomposition (SVD) is applied to $\widetilde{\bf{H}}_w$ to determine its eigenvalue. We assume that there is a single cluster in the scattering environment.
The center arriving angle of the cluster is set to $\tilde{\dot\theta}_{g}=\pi,g\in\left\{r,t\right\}$ with a fixed elevation angular spread of $180$ degrees (i.e., $\tilde{\phi}_{g}\in \tilde{\Omega}_{g}^{\tilde \phi}=\left[0,\pi\right]$). The channel coefficients of $\widetilde{\bf H}_{w}$ can be generated based on the CDL channel coefficient generation method. It is observed that the spatial DoF of the DARISA MIMO system increases with both the size of the DARISA array and the angular spread of the cluster. When the azimuth angular spread is set to $180^{\circ}$ (i.e., $c_{g,1}=c_{g,2}=1$), the array size is set to $\tilde{D}_{g,x}\tilde{D}_{g,y}=4\lambda\times 4\lambda$ or $8\lambda\times 8\lambda$, the DoF is equal to $55$ or $214$, which is consistent with the DoF theoretical analysis of Lemma 1 (i.e., $\tilde{d}_{g}\approx \pi\tilde{D}_{g,x}\tilde{D}_{g,y}$). In addition, for the two curves labeled ``$\tilde \Omega^{\tilde \theta}_{g} \in \left[0,180^{\circ}\right]$,  $\tilde{D}_{g,x}\tilde{D}_{g,y}=4\lambda\times4\lambda$, $N_g=16\times16$" and ``$\tilde \Omega^{\tilde \theta}_{g} \in \left[0,180^{\circ}\right]$,  $\tilde{D}_{g,x}\tilde{D}_{g,y}=4\lambda\times4\lambda$, $N_g=32\times 32$", the identical DoF values confirm that denser metasurface element deployment does not affect DoF when DARISA array size and angular spread remain constant.

	\begin{figure}
		\centering
		\begin{minipage}[t]{0.48\textwidth}
			\centering
			\includegraphics[width=3.8in]{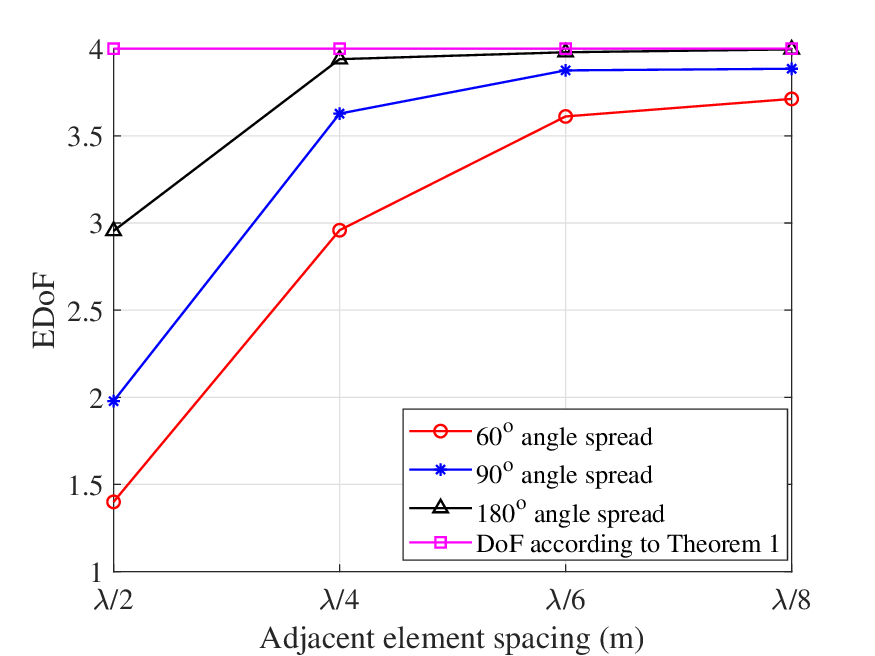}
			\caption{ The EDoF versus element spacing for different angular spreads by setting $M=4,N=2,K=2$, the size of each DARISA $\lambda\times \lambda$.}
		\end{minipage}
		\begin{minipage}[t]{0.48\textwidth}
			\centering
			\includegraphics[width=3.8in]{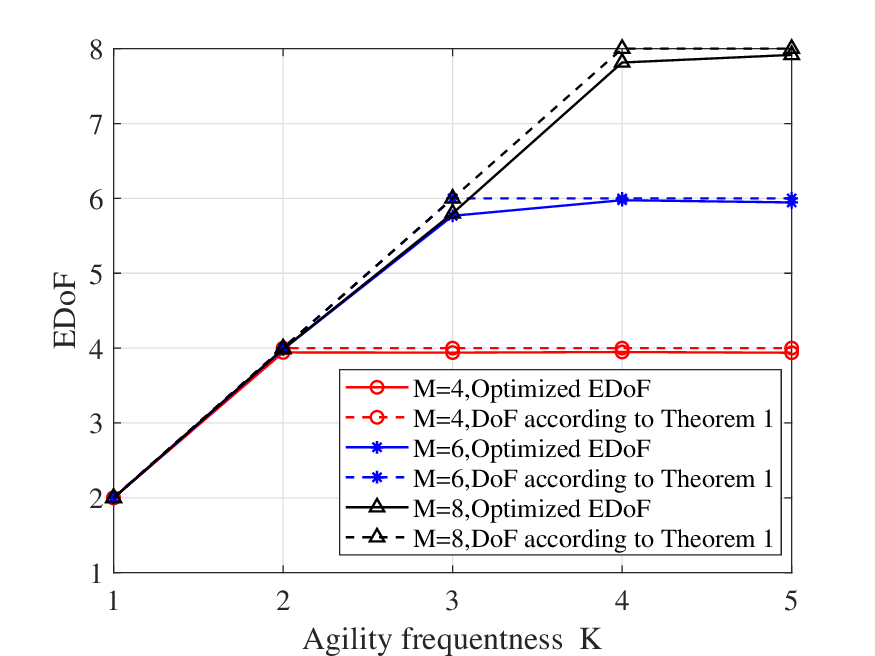}
			\caption{The EDoF versus agility frequentness $K$ for different $M$ by setting $N=2$, element spacing $\lambda/4$, and the size of each DARISA $\lambda\times \lambda$.}
		\end{minipage}
	\end{figure}
	
	Fig. 5 compares eigenvalues and channel capacity across different phase design schemes and cluster angular spreads for a randomly generated CDL channel, where $\tilde{\phi}_{g}\in \tilde{\Omega}_{g}^{\tilde \phi}=\left[0,\pi\right]$. The simulation curve generation method in Fig. 5(a) is the same as that in Fig. 4, including cluster parameter configuration, channel coefficient generation and channel matrix decomposition steps. When generating the coefficients of composite spatial-temporal channel matrix $\overline{\bf H}_C$ in (\ref{spatial-temporal-MIMO_channel}), the effects of $\overline{\bf{Q}}_{r}$ and $\overline{\bf{Q}}_{t}$
 need to be considered. We set the the number of cluster $\tilde{L}=1$, the agility frequentness $K=4$, the number of transmit and receive antennas $M=8, N=2$, the number of dense deployment elements of each DARISA $16\times16$, and the size of each DARISA ${2\lambda}\times{2\lambda}$. The random phase scheme serves as the benchmark. As observed in Fig. 5(a), the EDoF increases with angular spreads. This is due to the fact that the spatial DoF of $\overline{\bf H}_C$ in (\ref{spatial-temporal-MIMO_channel}) increases with the angular spread, so does the EDoF. While the number of eigenvalues remains the same for both the random phase scheme and our proposed optimized phase scheme, the phase responses of $\overline{\bf{Q}}_{r}$ in (\ref{spatial-temporal-MIMO_channel}) optimized by our proposed algorithm can make all singular values of $\overline{\mathbf{H}}_C$ approximately equal, leading to a significant enhancement in EDoF. As shown in Fig. 5(b), for azimuth angular spreads of $30^{\circ},90^{\circ}$ and $180^{\circ}$, the EDoFs increase from  $1.08$, $3.21$ and $3.56$ with random phases to $4.21$, $5.87$ and $7.61$ with optimized phases. However, the EDoF cannot exceed $\textrm{rank}\left(\overline{\mathbf{H}}_C\right)= \min\left(KN, M,\tilde{d}_t,\tilde{d}_r\right)$.

	Fig. 6 depicts the EDoF versus element spacing for different angular spreads with transmit-DARISAs $M=4$, receive-DARISAs $N=2$, $K=2$ agility frequentness, each DARISA sized at $\lambda\times \lambda$, and the number of cluster $\tilde{L}=1$. The purple square curve denotes analytical results derived from Theorem 1, while other curves represent simulations averaged over 100 random CDL channel realizations. An important observation is that maintaining a constant DARISA array size while deploying ``denser" elements can enhance the EDoF. Because increasing $N_r$ elevates the dimensions of $\overline{\bf{Q}}_{r}\in \mathbb{C}^{KNN_r\times KN}$ in (\ref{spatial-temporal-MIMO_channel}). Consequently, more phase responses of elements can be designed, enabling enhanced optimization of $\overline{\bf{Q}}_{r}$ to achieve more uniform singular value distributions in $\overline{\mathbf{H}}_C$. However, EDoF growth rate diminishes at smaller spacings due to: 1) Channel rank constraint: EDoF is fundamentally bounded by $\mathrm{rank}(\overline{\mathbf{H}}_C) = \min(KN, M, \tilde{d}_t, \tilde{d}_r)$. As spacing decreases, EDoF approaches this theoretical upper bound, yielding diminishing returns; 2) Optimization marginality: Incremental elements yield diminishing returns in beamforming flexibility and singular value equalization beyond critical density. Additionally, EDoF increases with angular spread, consistent with Fig. 5 observations.
	
	We show the EDoF versus agility frequentness $K$ for various numbers of transmit-DARISAs $M$ in an isotropic scattering environment in Fig. 7. Analytical results derived from Theorem 1 are denoted by dotted curves, while solid curves correspond to Monte Carlo simulations averaged over 100 random CDL channel realizations under isotropic scattering conditions. We set receive-DARISs $N=2$, $\lambda/4$ element spacing, and the size of each DARISA at $\lambda\times \lambda$. The results validate that the EDoF increases with $K$, as the dimensions of $\overline{\bf{Q}}_{r}\in \mathbb{C}^{KNN_r\times KN}$ in (\ref{spatial-temporal-MIMO_channel}) grow with $K$,
	This increase allows for more phase shifts to be artificially designed, leading to singular values of $\overline{\mathbf{H}}_C$ that are more closely aligned. Notably, EDoF asymptotically approaches but cannot surpass the fundamental limit $\textrm{rank}\left(\overline{\mathbf{H}}_C\right)= \min\left(KN, M,\tilde{d}_t,\tilde{d}_r\right)$, aligning with Theorem 1.
	
	\begin{figure}
		\centering
		{\includegraphics[width=3.8in]{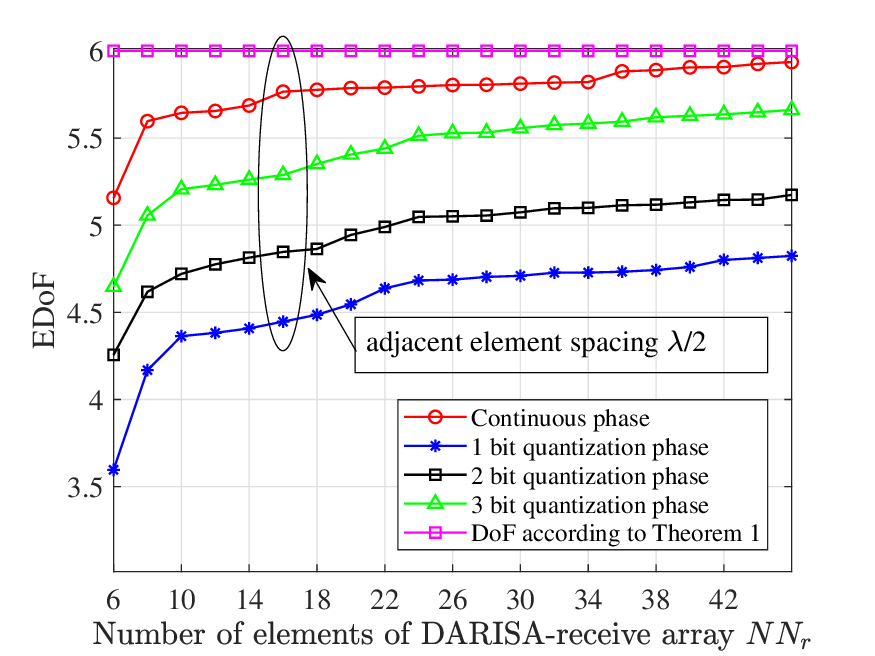}}
		\caption{The EDoF and DoF versus number of receive elements for different quantized phase bits by setting $M=6,N=2,K=3$, the size of receive DARISA array $2\lambda\times2\lambda$.}
	\end{figure}
	
	Fig. 8 illustrates the EDoF and DoF of composite spatial-temporal channel $\overline{\bf H}_C$  as functions of the number of receive metasurface elements under continuous and discrete phase responses. System parameters include
	 $M=6$ transmit-DARISAs,  $N=2$ receive-DARISAs, $K=3$ agility frequentness, and $2\lambda\times2\lambda$ array size of receive DARISA.  Analytical results from Theorem 1 are depicted by purple square markers, while simulation curves represent Monte Carlo averages over 100 random CDL channel realizations in an isotropic scattering environment. The results indicate that by fixing the DARISA array size, the EDoF increases with both quantized phase bits and the number of dense deployment elements. This occurs because: 1) increasing the number of dense deployment elements enhances the dimensions of the adjustable phase response, allowing for more phase shifts to be artificially designed; and 2) a higher phase quantization bits improves the accuracy of constructive superposition of $\overline{\bf H}_C$. Both factors contribute to making the channel well-conditioned, thereby improving the EDoF. These findings suggest that deploying a denser arrangement of elements can mitigate performance losses associated with lower phase accuracy.
	
	\section{Conclusion}
In this paper, we proposed a DARISA MIMO system that utilized the ability of metasurface elements to agilely adjust the phase response. We analyzed the DoFs of the DARISA MIMO system and derived an explicit relationship between DoFs and system parameters including agility frequentness, angular spread, DARISA array size, and the number of DARISAs. This finding suggested that it was possible to increase the DoFs when the number of receive-DARISAs was less than the number of transmit-DARISAs. Then, the EDoF was used to measure capacity and we optimized the agile phase responses of metasurface elements to maximize the EDoF. To solve this non-convex problem, fractional programming and SDR algorithms were developed. Simulation results validated the effectiveness of our proposed method, demonstrating that denser deployment of DARISA elements enhanced the EDoF and that a more densely arranged deployment of elements could compensate for the loss of communication performance due to low phase quantization accuracy. The DARISA MIMO system showed significant potential for broad applications. For examples, 1) environment sensing and channel estimation leveraging DARISA’s dynamic reconfigurability; 2) adaptive beamforming optimization in multi-DARISA networks to maximize the sum rate; 3) artificial channel randomness generation through agile adjustment of DARISA phases to enhance physical layer security and so on.

	{\appendices
		\section{Proof of Lemme 1} \label{proof1}
	{ The proof comprises two main steps: 1) establishing the equivalence between the space-domain channel ($\tilde{a}$ superscripted notation) and the wavenumber-domain channel ($\breve{a}$ superscripted notation); 2) analyzing the DoF of the wavenumber-domain channel in both isotropic and non-isotropic scattering environments. Identical notations in $\tilde{a}$ and $\breve{a}$ retain consistent physical interpretations, with superscripts exclusively distinguishing their domain-specific formulations.
		
		\subsection{Equivalent Channel from Space to Wavenumber Domain}
		
		According to (\ref{Scatter_MIMO_channel}), the DARISA MIMO spatial array-scattering response $\widetilde{\bf{H}}_w\in{{\mathbb{C}}^{NN_r\times MN_t}}$ is expressed as
	\begin{align}
		\tag{A-1}\label{Scatter_MIMO}
		\widetilde{\bf{H}}_w=\frac{1}{\sqrt{\tilde{L}}}\sum\limits_{\tilde{l}=1}^{\tilde{L}}\widetilde{\bf A}_{r}^{\tilde{l}H}\widetilde{\bf H}_{a}^{\tilde{l}}\widetilde{\bf A}_{t}^{\tilde{l}}=\frac{1}{\sqrt{\tilde L}}\sum\limits_{\tilde l=1}^{\tilde L}\widetilde{\bf H}_{w}^{\tilde l},
	\end{align}
	where the meaning of all notations refers to (\ref{Scatter_MIMO_channel}). The spatial array-scattering channel response $\tilde{h}_{uv},u\in\mathcal{U}=\left\{1,\cdots,NN_r\right\}, v\in\mathcal{V}=\left\{1,\cdots,MN_t\right\}$ between $u$-th receive and $v$-th transmit element is expressed as
		\begin{align}
			\notag
			\tilde{h}_{uv}=\frac{1}{\sqrt{\tilde L}}\sum_{\tilde l=1}^{\tilde L}\  \iint\limits_{{{\tilde \Omega}}_{t}^{\tilde l} \times {{\tilde \Omega}}_{r}^{\tilde l} }& {a}_{r}^{u}\left( {\boldsymbol{\alpha} }^{\tilde l}_{r}\right){\tilde H}_{a}\left( {\boldsymbol{\alpha}}^{\tilde l}_{r}, {\boldsymbol{\alpha} }^{\tilde l}_t\right)\\
		\tag{A-2}	\label{space-countinue}
			&	{a}_{t}^{v*}\left({\boldsymbol{\alpha}}^{\tilde l}_t \right)d {\boldsymbol{\alpha}}^{\tilde l}_t d,
		\end{align}
		where ${\tilde {\Omega}}_{g}^{\tilde l}= \Omega_{g,\tilde{\theta}}^{\tilde l}\times \Omega_{g,\tilde{\phi}}^{\tilde l},g\in\left\{r,t\right\}$. The meaning of other notations refers to (2).
		
		Based on the works \cite{Pizzo4}\cite{Zhang23}, for an electromagnetic wave with spatial azimuth angle $\phi$ and elevation angle $\theta$,
		the wave propagation direction in the wavenumber domain is a scaled version of the direction in the space domain, i.e.,  $\phi$ and $\theta$ can be mapped to the wavenumber domain as
		\begin{align}
			\notag
				\kappa_x=\beta \sin \theta \cos \phi,
				\kappa_y=\beta \sin \theta \sin \phi, 
				\kappa_z=\beta \cos \theta,
		\end{align}
		where $\beta=\frac{2\pi}{\lambda}$ is the wavenumber, $	\kappa_x, \kappa_y$ and $\kappa_z$ denote the
		wavenumber in $x, y$ and $z$ directions, respectively. Thus, by defining the space-domain propagation direction sets as ${\tilde {\Omega}}_{g}= {\tilde {\Omega}}_{g}^{1}\bigcup,\cdots,\bigcup {\tilde {\Omega}}_{g}^{\tilde L},g\in\left\{r,t\right\}$ and the wavenumber-domain departure/arrival direction vector as $\boldsymbol{\kappa}_g=[\kappa_{g,x},  \kappa_{g,y},\kappa_{g,z}]^T$,  the departure/arrival directions set in wavenumber domain $\mathcal{D}(\boldsymbol{\kappa}_g) $
	 mapped from ${\tilde {\Omega}}_{g}$ is expressed as
		\begin{align}
			\notag
		\mathcal{D}({\boldsymbol \kappa}_g)=\{({\kappa}_{g,x},{\kappa}_{g,y},{{\kappa}_{g,y}})\in {{\mathbb{R}}^{3}}:\kappa_{g,x}^{2}+\kappa_{g,y}^{2}+\kappa_{g,y}^{2}= {{\beta }^{2}}\}.
		\end{align}

		Since $\kappa_{g,x}^{2}+\kappa_{g,y}^{2}+\kappa_{g,y}^{2}= {{\beta }^{2}}$, the departure/arrival direction set in the wavenumber domain
		is limited to the support $$\mathcal{D}({\boldsymbol \kappa}_g)=\{({\kappa}_{g,x},{\kappa}_{g,y})\in {{\mathbb{R}}^{2}}:\kappa_{g,x}^{2}+\kappa_{g,y}^{2}\leq {{\beta }^{2}}\},\boldsymbol{\kappa}_g \in \mathcal{D}(\boldsymbol{\kappa}_g).$$
		
		Therefore, the spatial array-scattering channel (\ref{space-countinue}) can be equivalently mapped to the wavenumber domain as
		\begin{align}
			\notag
		&	\tilde{h}_{uv}= \iint\limits_{ \mathcal{D}(\boldsymbol{\kappa}_r) \times \mathcal{D}(\boldsymbol{\kappa}_t)} \breve{a}_{r}^{u}\left( {\boldsymbol{\kappa} }_{r}\right){\breve H}_{a}\left( {\boldsymbol{\kappa}}_{r}, {\boldsymbol{\kappa} }_t\right)\breve{a}_{t}^{v*}\left({\boldsymbol{\kappa}}_t \right)d {\boldsymbol{\kappa}}_t d {\boldsymbol{\kappa}}_r,\\
			\tag{A-3}\label{wavenumber-con}
		&	u\in\mathcal{U}=\left\{1,\cdots,NN_r\right\}, v\in\mathcal{V}=\left\{1,\cdots,MN_t\right\},
		\end{align}
		where $\breve{a}_{t}^{v}\left({\boldsymbol{\kappa}}_t \right)={e}^{\jmath{\boldsymbol{\kappa}}_t^{T}{\bf{t}}_v}$ is the \emph{transmit response} that maps the impulsive excitation current ${\bf t}_v$ at the location of $v$-th element ${\bf{t}}_v$ to the departure propagation direction ${\boldsymbol{\kappa}}_t$,  $\breve{a}_{r}^{u}\left({\boldsymbol{\kappa}}_r \right)={e}^{\jmath{\boldsymbol{\kappa}}_r^{T}{\bf{r}}_u}$ is the \emph{receive response} that maps the arrival propagation direction ${\boldsymbol{\kappa}}_r$ to the
		induced current at the location of $u$-th element ${\bf{r}}_u$, ${\breve H}_{a}\left( {\boldsymbol{\kappa}}_{r}, {\boldsymbol{\kappa} }_t\right)\sim \mathcal{CN}(0, {\breve \sigma}^2\left( {\boldsymbol{\kappa}}_{r}, {\boldsymbol{\kappa} }_t\right))$ is the \emph{angular response} that maps departure direction ${\boldsymbol{\kappa} }_t$ to arrival direction ${\boldsymbol{\kappa}}_{r}$, ${\breve \sigma}^2\left( {\boldsymbol{\kappa}}_{r}, {\boldsymbol{\kappa} }_t\right)$ represents the power of angular response.
		
		Then, based on [11, Theorem 2], the above continuous plane wave channel (\ref{wavenumber-con}) is a bandwidth-constrained channel $|\mathcal{D}({\boldsymbol \kappa}_g)|=\pi\beta^2,g\in\left\{r,t\right\}$. When the sampling interval satisfies the Nyquist condition, this continuous channel can be recovered from a finite number of discrete samples. Sampling $\tilde{h}_{uv}$ at intervals of $\Delta_{{\boldsymbol{\kappa}}_{g,x}}=\frac{2\pi}{{\tilde D }_{g,x}},\Delta_{{\boldsymbol{\kappa}}_{g,y}}=\frac{2\pi}{{\tilde D }_{g,y}}$ in the wavenumber domain, (\ref{wavenumber-con}) is transformed as
		\begin{align}
			\notag
			\tilde{h}_{uv}&= \sum\nolimits_{(\breve{p}_{r,x},\breve{p}_{r,y})\in \breve{\mathcal{E}_{r}}} \sum\nolimits_{(\breve{p}_{t,x},\breve{p}_{t,y})\in \breve{\mathcal{E}_{t}}}\breve{a}_{r}^{u}\left( \breve{p}_{r,x},\breve{p}_{r,y}\right)
			\\
			\tag{A-4}	\label{wavenumber-dis}
			& \qquad \qquad  {\breve H}_{a}\left(\breve{p}_{r,x},\breve{p}_{r,y},\breve{p}_{t,x},\breve{p}_{t,y}\right)\breve{a}_{t}^{v*}\left(\breve{p}_{t,x},\breve{p}_{t,y}\right),
		\end{align}
		where
		$$
		\breve{\mathcal{E}_{g}}=\left\{ (\breve{p}_{g,x},\breve{p}_{g,y})\in {\mathbb{Z}}^{2}:\left(\frac{\breve{p}_{g,x}}{ {\breve D}_{g,x}}\right)^2+\left(\frac{\breve{p}_{g,y}}{{\breve D}_{g,y}}\right)^{2}\le 1\right\},
		$$ is the discrete wavenumber departure/arrival directions set, $\breve{a}_{r}^{u}\left( \breve{p}_{r,x},\breve{p}_{r,y}\right)=e^{\jmath\left(\frac{2 \pi \breve{p}_{r,x}}{\breve{D}_{r,x}} r_{x}^{u}+\frac{2 \pi\breve{p}_{r,y}}{\breve{D}_{r,y}} r_{y}^{u}\right)}$ is recieve response, $\breve{a}_{t}^{v}\left( \breve{p}_{t,x},\breve{p}_{t,y}\right)=e^{\jmath\left(\frac{2 \pi \breve{p}_{t,x}}{\breve{D}_{t,x}} t_{x}^{v}+\frac{2 \pi\breve{p}_{t,y}}{\breve{D}_{t,y}} t_{y}^{v}\right)}$ is transmit response, ${\breve H}_{a}\left(\breve{p}_{r,x},\breve{p}_{r,y},\breve{p}_{t,x},\breve{p}_{t,y}\right)\sim \mathcal{CN}(0, {\breve \sigma}^2\left(\breve{p}_{r,x},\breve{p}_{r,y},\breve{p}_{t,x},\breve{p}_{t,y}\right))$ is the angular response, where
	\begin{align}
		\notag
		 &{\breve \sigma}^2\left(\breve{p}_{r,x},\breve{p}_{r,y},\breve{p}_{t,x},\breve{p}_{t,y}\right)=\\
		 \tag{A-5}\label{power}
		 &\iiiint\limits_{\breve{W}}{\breve S} \left(\breve{p}_{r,x},\breve{p}_{r,y},\breve{p}_{t,x},\breve{p}_{t,y}\right)d\breve{p}_{r,x}d\breve{p}_{r,y} d\breve{p}_{t,x} d\breve{p}_{t,y},
	\end{align}
$	{\breve S} \left(\breve{p}_{r,x},\breve{p}_{r,y},\breve{p}_{t,x},\breve{p}_{t,y}\right)$ represents the power spectral density in wavenumber domain, the integration region $\breve{W}$ is expressed as
	\begin{align}
		\notag
	\breve{W}={{W}_{r}(\breve{p}_{r,x},\breve{p}_{r,y})\times {W}_{t}(\breve{p}_{t,x},\breve{p}_{t,y})},
		\end{align}
	where ${{W}_{g}}(\breve{p}_{g,x},\breve{p}_{g,y}),g\in\left\{r,t\right\}$ is expressed as
		\begin{align}
	\notag
	&	\left\{\left[\frac{2\pi \breve{p}_{g,x}}{{\breve D}_{g,x}},\frac{2\pi \left(\breve{p}_{g,x}+1\right)}{{\breve D}_{g,x}}\right]\times \left[\frac{2\pi \breve{p}_{g,y}}{{\breve D}_{g,y}},\frac{2\pi \left(\breve{p}_{g,y}+1\right)}{{\breve D}_{g,y}}\right]\right\}.
		\end{align}
Therefore,  $\widetilde{\bf H}_{w}\in \mathbb{C}^{NN_r \times MN_t }$ in (\ref{Scatter_MIMO}) can be equivalently transformed in the wavenumber domain \cite{Pizzo2}, i.e.,
	\begin{align}
		\tag{A-6}\label{wave-MIMO}
		\breve{\bf H}_{w} = \sqrt{NN_rMN_t} \breve{\bf A}_r  \breve{\bf H}_a \breve{\bf A}_t^H,
	\end{align}
where $\breve{\bf A}_r \in \mathbb{C}^{NN_r \times \breve{d}_r}$ and $\breve{\bf A}_t \in \mathbb{C}^{MN_t \times \breve{d}_t}$ collect column vectors ${\bf \breve{a}}_{r}\left(\breve{p}_{r,x},\breve{p}_{r,y}\right) \in \mathbb{C}^{NN_r \times 1}$ and ${\bf \breve{a}}_{t}\left(\breve{p}_{t,x},\breve{p}_{t,y}\right) \in \mathbb{C}^{MN_t \times 1}$ which are the normalized array responses with entries $\left[{\bf \breve{a}}_{r}\left(\breve{p}_{r,x},\breve{p}_{r,y}\right) \right]_u=\frac{1}{\sqrt{NN_r}}{\breve{a}}^{u}_{r}\left(\breve{p}_{r,x},\breve{p}_{r,y}\right), u\in\mathcal{U},$ and $\left[{\bf \breve{a}}_{t}\left(\breve{p}_{t,x},\breve{p}_{t,y}\right) \right]_v=\frac{1}{\sqrt{MN_t}}{\breve{a}}^{v}_{t}\left(\breve{p}_{t,x},\breve{p}_{t,y}\right), v\in\mathcal{V}$. These are semi-unitary matrices, i.e., $\breve{\bf A}_r^H\breve{\bf A}_r={\bf I}_{\breve{d}_r}, \breve{\bf A}_t^H\breve{\bf A}_t={\bf I}_{\breve{d}_t}$. $\breve{d}_g, g\in\left\{r,t\right\}$ is the cardinality $|\breve{\mathcal{E}_{g}}|$ of the departure/arrival directions in the wavenumber domain. $\breve{\bf H}_a \in \mathbb{C}^{\breve{d}_r \times \breve{d}_t}$ is the angular random matrix collecting ${\breve H}_{a}\left(\breve{p}_{r,x},\breve{p}_{r,y},\breve{p}_{t,x},\breve{p}_{t,y}\right)$.
	
Based on the above analysis, we establish the equivalence of spatial array-scattering response $\widetilde{\bf H}_{w}$ in (\ref{Scatter_MIMO}) and the wavenumber array-angular response  $\breve{\bf H}_{w}$ in  (\ref{wave-MIMO}). This implies that $rank\left(\widetilde{\bf H}_{w}\right)=rank\left(\breve{\bf H}_{w}\right)$. Next, we analyze $rank\left(\breve{\bf H}_{w}\right)$.
	
\subsection{ DoF Analysis in Isotropic and Non-isotropic Scattering}
	First, $\breve{\bf H}_{w}$ in (\ref{wave-MIMO}) is re-expressed as
	\begin{align}
		\tag{A-7}\label{wave-MIMO-further}
		\breve{\bf H}_{w} =\sqrt{NN_rMN_t}\breve{\bf A}_r  \left(\breve{\boldsymbol{\Sigma}}\odot\breve{\boldsymbol{\Xi}}\right) \breve{\bf A}_t^H,
	\end{align}
where $\breve{\boldsymbol{\Sigma}}\in \mathbb{C}^{\breve{d}_r \times \breve{d}_t}$ represents the matrix consisting of the standard deviations ${\breve \sigma}\left(\breve{p}_{r,x},\breve{p}_{r,y},\breve{p}_{t,x},\breve{p}_{t,y}\right)$ of angular response ${\breve H}_{a}\left(\breve{p}_{r,x},\breve{p}_{r,y},\breve{p}_{t,x},\breve{p}_{t,y}\right)$ in all wavenumber departure/arrival directions. $\breve{\boldsymbol{\Xi}}\in \mathbb{C}^{\breve{d}_r \times \breve{d}_t}\sim \mathcal{CN}\left({\bf 0},{\bf I}_{\breve{d}_r\breve{d}_t}\right)$ represents the channel randomness.

Since $\breve{\bf A}_r^H\breve{\bf A}_r={\bf I}_{\breve{d}_r}, \breve{\bf A}_t^H\breve{\bf A}_t={\bf I}_{\breve{d}_t}$, the non-zero eigenvalues of ${{\bf \breve{H}}_w}$ are determined by $\breve{\boldsymbol{\Sigma}}$ and $\breve {\boldsymbol{\Xi}}$. Therefore, the number of non-zero eigenvalues of $\breve{\bf H}_{w}\in \mathbb{C}^{NN_r \times MN_t }$ and $\breve{\bf H}_{a}=\left(\breve{\boldsymbol{\Sigma}}\odot\breve{\boldsymbol{\Xi}}\right) \in \mathbb{C}^{\breve{d}_r\times \breve{d}_t }$ are identical, so the ranks of them are identical, i.e.,
	\begin{align}
	\tag{A-8}\label{rank}
	rank\left({\breve{\bf H}_{w}}\right) =	rank\left({\breve{\bf H}_{a}}\right).
\end{align}
The rank of $\breve{\bf H}_{a}$  is determined by the number of non-zero rows $\breve{d}_r^{'}$ and non-zero columns $\breve{d}_t^{'}$
, where  $rank\left({\breve{\bf H}_{a}}\right)=\min\left\{\breve{d}_r^{'},\breve{d}_t^{'}\right\}$.

According to \cite{Pizzo2}, $\breve{d}_r^{'}$ and $\breve{d}_t^{'}$ are jointly determined by the array size and wireless scattering environment. The array size determine the cardinality $\breve{d}_g=|\breve{\mathcal{E}_{g}}|$ of the sets $(\breve{p}_{g,x},\breve{p}_{g,y})\in \breve{\mathcal{E}_{g}}$, i.e.,
	\begin{align}
	\tag{A-9}\label{set}
	\breve{d}_g=|\breve{\mathcal{E}_{g}}|\approx{\pi}\tilde{D}_{g,x} \tilde{D}_{g,y}, g\in\left\{r,t\right\}.
\end{align}
The wireless scattering environment governs the distrbution of ${\breve \sigma}^2\left(\breve{p}_{r,x},\breve{p}_{r,y},\breve{p}_{t,x},\breve{p}_{t,y}\right)$. The number of path clusters and their angular spreads directly influence the distribution of ${\breve \sigma}^2\left(\breve{p}_{r,x},\breve{p}_{r,y},\breve{p}_{t,x},\breve{p}_{t,y}\right)$.  Since the power spectral density in the wavenumber domain (\ref{power}) is a nonlinear function of the number of path clusters and angular spreads, a closed-form solution for $\breve{d}_r^{'}$ and $\breve{d}_t^{'}$ in terms of these variables cannot be directly derived. To further analyze $\breve{d}_r^{'}$ and $\breve{d}_t^{'}$, the isotropic and non-isotropic scattering environments are considered, respectively.

1) In the isotropic environment with rich scattering, signals are incident from any direction and randomly scattered to any direction. Therefore, the angular power of any propagation direction in the wavenumber domain ${\breve \sigma}^2\left(\breve{p}_{r,x},\breve{p}_{r,y},\breve{p}_{t,x},\breve{p}_{t,y}\right)$ is non-zero. This implise that the $\breve{d}_r^{'}$ and $\breve{d}_t^{'}$ are only determined by the array size, i.e.,
	\begin{align}
		\tag{A-10}
		\breve{d}_g^{'}=\breve{d}_g\approx \pi\tilde{D}_{g,x} \tilde{D}_{g,y}, g\in\left\{r,t\right\}.
\end{align}
Therefore, the rank of $\breve{\bf H}_{a}$ is expressed as
	\begin{align}
\tag{A-11} \label{iso-dof}
	rank\left({\breve{\bf H}_{a}}\right)= \min\left\{\breve{d}_r,\breve{d}_t\right\}.
\end{align}

\begin{figure}[!t]
	\centerline{\includegraphics[width=3.4in]{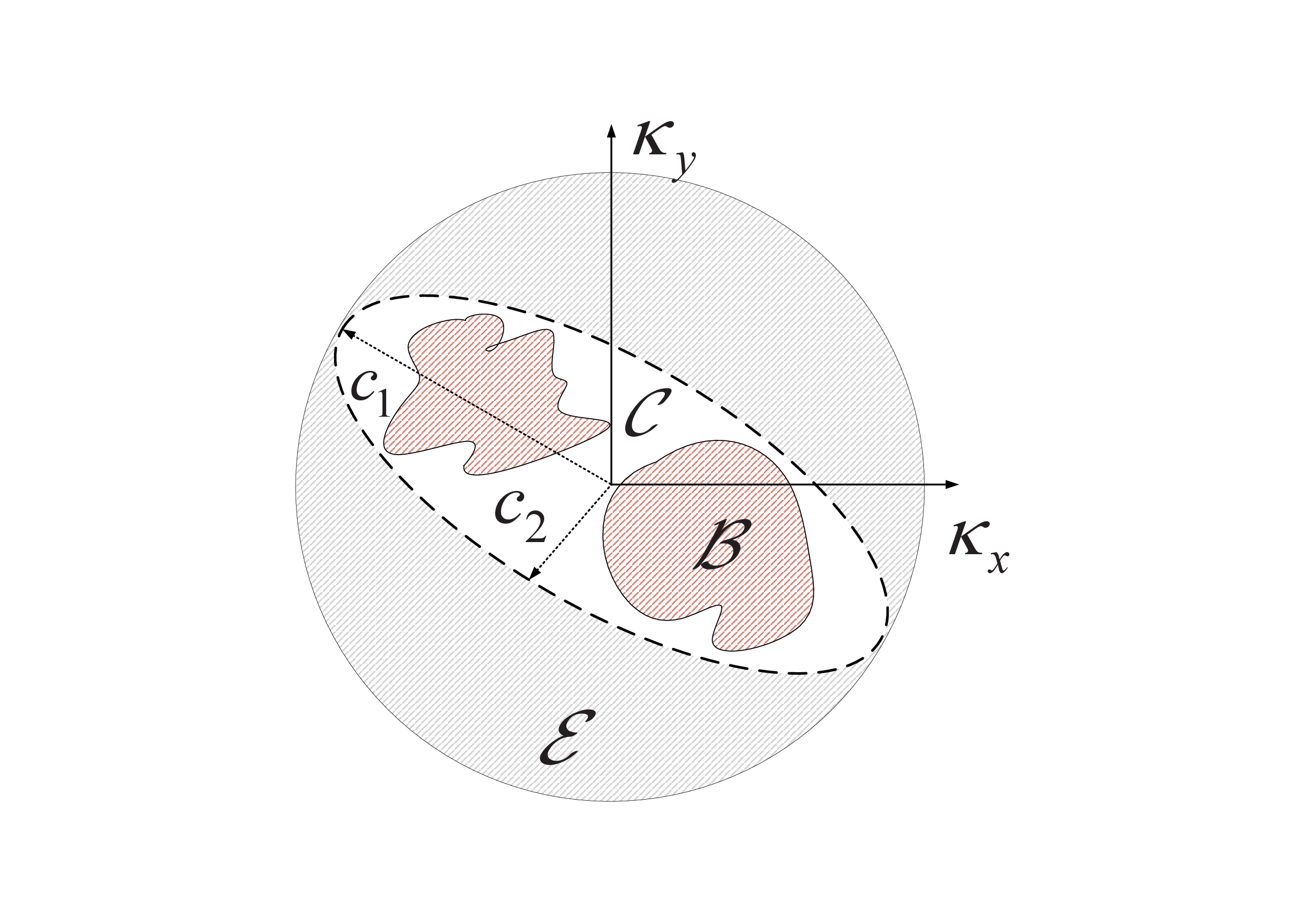}}
	\caption{Propagation in isotropic and non-isotropic environments.}
	\vspace{-0.6cm}
\end{figure}

2) In the non-isotropic scattering environment,  according to \cite{Pizzo3}, the blockage effect in a non-isotropic scattering environment restricts the electromagnetic wave propagation directions in the wavenumber domain to a smaller support $ \breve{\mathcal{B}_{g}}	\subseteq \breve{\mathcal{E}_{g}}, g\in\left\{r,t\right\}$, as illustrated in Fig. 9. Consequently, the scattering mechanism is transformed into another filtering operation based on migration filtering. However, deriving the Nyquist period matrix under arbitrary non-isotropic conditions is typically intractable due to its complexity. To address this challenge, a suboptimal strategy can be pursued by embedding $\breve{\mathcal{B}_{g}}$ into a larger and connected support $\breve{\mathcal{B}_{g}}\subseteq\breve{\mathcal{C}_{g}}\subseteq \breve{\mathcal{E}_{g}}$ [14, Section IV.C] depicted in Fig. 9, i.e.,
$$
\breve{\mathcal{C}_{g}}=\left\{ (\breve{p}_{g,x},\breve{p}_{g,y}):\left(\frac{\breve{p}_{g,x}}{ {{{c}}_{g,1}\breve D}_{g,x}}\right)^2+\left(\frac{\breve{p}_{g,y}}{{{c}}_{g,2}{\breve D}_{g,y}}\right)^{2}\le 1\right\},
$$
where $0\leq {c}_{g,1}\leq 1$ and $0\leq {c}_{g,2}\leq 1$ represent the semi-major and semi-minor axes, respectively, of the tightest ellipse encompassing all space-domain clusters projected onto the wavenumber domain, which are proportional to the cluster angle spread.

According to [14, Section V], the upper bound of $\breve{d}_r^{'}$ and $\breve{d}_t^{'}$ is expressed as
	\begin{align}
		\tag{A-12}
	\breve{d}_g^{'}\approx{c}_{g,1}{c}_{g,2}\pi\tilde{D}_{g,x} \tilde{D}_{g,y}, g\in\left\{r,t\right\}.
\end{align}
Therefore, the rank of $\breve{\bf H}_{a}$ is expressed as
\begin{align}
		\tag{A-13}\label{noniso-dof}
	rank\left({\breve{\bf H}_{a}}\right)= \min\left\{\breve{d}_r^{'},\breve{d}_t^{'}\right\}.
\end{align}

Since the spatial array-scattering response $\widetilde{\bf H}_{w}$ in (\ref{Scatter_MIMO}) and the wavenumber array-anguler response  $\breve{\bf H}_{w}$ in  (\ref{wave-MIMO}) are equivalent, i.e., $rank\left(\widetilde{\bf H}_{w}\right)=rank\left(\breve{\bf H}_{w}\right)$, combined consideration of (\ref{rank}), i.e., $rank\left(\breve{\bf H}_{w}\right)=rank\left(\breve{\bf H}_{a}\right)$, the DoF of $\widetilde{\bf H}_{w}$ can be futher expressed as (the DoF is the rank of channel matrix)
 \begin{align}
 \tag{A-14}	\label{space-dof}
 rank\left(\widetilde{\bf H}_{w}\right)= \min\left\{\tilde{d}_r,\tilde{d}_t\right\},
 \end{align}
 where $\tilde{d}_g\approx{c}_{g,1}{c}_{g,2}\pi\tilde{D}_{g,x} \tilde{D}_{g,y}, {c}_{g,1},{c}_{g,2}\in\left[0, 1\right], g\in\left\{r,t\right\}, $
${c}_{g,1}$ and ${c}_{g,2}$ is determined by the number of clusters and their angular spreads. When the scattering environment is isotropic, ${c}_{g,1}={c}_{g,2}=1$.}
			
		\section{ Proof of Theorem 1}\label{proof-of-theorem1}
		According to Lemma 1 that $\textrm{rank}\left(\widetilde{\bf H}_w\right)=\min\left(\tilde{d}_t,\tilde{d}_r\right)$, we obtain $\textrm{rank}\left(\overline{\bf H}_w\right)=\min\left(K\tilde{d}_t,K\tilde{d}_r\right)$. From matrix multiplication, we know that $\textrm{rank}\left(\overline{\bf H}_C\right)\leq \min\left(\textrm{rank}(\overline{\bf H}_w),\textrm{rank}\left({\overline{\bf{Q}}}_{r}\right),\textrm{rank}\left({\overline{\bf{Q}}}_{t}\right)\right)\leq \min\left(KN,M,K\tilde{d}_t,K\tilde{d}_r\right)$. Since $\textrm{rank}\left(\overline{\bf H}_C\overline{\bf H}_C^H\right)=\textrm{rank}\left(\overline{\bf H}_C^H \overline{\bf H}_C\right)=\textrm{rank}\left(\overline{\bf H}_C\right)$, analyzing the rank of $\overline{\bf H}_C \overline{\bf H}_C^H$ is equivalent to analysing the rank of $\overline{\bf H}_C$. ${\overline{\bf{Q}}}_{r}$ and ${\overline{\bf{Q}}}_{t}$ can be artificially designed to be column-full rank, \begin{align}
			\notag
			&{\overline{\bf{Q}}}_{r}{\overline{\bf{Q}}}_{r}^H={\bf U}_r\textnormal{diag}\left\{{{\boldsymbol{\Lambda}}_r},{\bf 0}_{KNN_r-KN}\right\}{\bf V}_r^H,\\
			\notag
			&{\overline{\bf{Q}}}_{t}{\overline{\bf{Q}}}_{t}^H={{\bf U}_t}\textnormal{diag}\left\{{{\boldsymbol{\Lambda }}_t},{\bf 0}_{KMN_t-M}\right\}{\bf V}_t^H,
		\end{align}
		where ${\bf U}_t\in \mathbb{C}^{KMN_t\times KMN_t},{\bf U}_r\in \mathbb{C}^{KNN_r\times KNN_r},{\bf V}_t\in \mathbb{C}^{KMN_t\times KMN_t}$, and ${\bf V}_r\in \mathbb{C}^{KNN_r\times KNN_r}$ are unitary matrices, ${\boldsymbol{\Lambda}}_t\in \mathbb{C}^{M\times M},{\boldsymbol{\Lambda}}_r\in \mathbb{C}^{KN\times KN}$ are full rank diagonal matrices, so we obtain
		\begin{align}
			\notag
			\overline{\bf H}_C \overline{\bf H}_C^H&={\overline{\bf{Q}}}_{r}^H{\overline{\bf H}_w} {{\bf U}_t}\textnormal{diag}\left\{{{\boldsymbol{\Lambda }}_t},{\bf 0}_{KMN_t-M}\right\}{\bf V}_t^H{\overline{\bf H}_w^H}{\overline{\bf{Q}}}_{r},\\
			\notag
			\overline{\bf H}_C^H \overline{\bf H}_C&={\overline{\bf{Q}}}_{t}^H{\overline{\bf H}_w^H} {\bf U}_r\textnormal{diag}\left\{{{\boldsymbol{\Lambda}}_r},{\bf 0}_{KNN_r-KN}\right\}{\bf V}_r^H {\overline{\bf H}_w}{\overline{\bf{Q}}}_{t}.
		\end{align}
		Therefore, by designing ${\overline{\bf{Q}}}_{r}$ and ${\overline{\bf{Q}}}_{t}$ to have full column rank, we  form a block matrix consisting of any
		$M$ columns of $\overline{\bf H}_w$ or any $KN$ rows of $\overline{\bf H}_w$. Thus, the rank of the matrix $\overline{\bf H }_C$ must satisfy $\textrm{rank}\left(\overline{\bf H}_C\right)\leq \min\left(KN,M \right)$.
		
		The rank of any $M$-column and $KN$-row block matrixes of $\overline{\bf H}_w$ are discussed below. According to structure of $\overline{\bf H}_C$ and $\overline{\bf H}_w$ in (\ref{spatial-temporal-MIMO_channel}), any $M$-column block matrix is derived from
		$\widetilde{\bf H}_w\in \mathbb{C}^{NN_r\times MN_s}$. When the $KN$-row block matrix of $ \overline{\bf H}_w$ exceeds the dimensions of $\widetilde{\bf H}_w$, the rank of $\overline{\bf H}_C$ is determined by the column block matrix; thus, it is sufficient to analyze the rank of the $KN$-row block matrix when $KN\leq M$ and the rank of the $M$-column block matrix of $\widetilde{\bf H}_w$.
		
		1) When $M\geq KN$, the rank of the composite spatial-temporal channel matrix is limited to $KN$ by the $KN$ rows of the block matrix of $\widetilde{\bf H}_w$. We have $\textrm{rank}\left(\widetilde{\bf H}_w\right)=\min\left(\tilde{d}_t,\tilde{d}_r\right)$. When $KN\leq \min\left(\tilde{d}_t,\tilde{d}_r\right)$, the rank of the channel is less than or equal to $KN$. Since ${\overline{\bf{Q}}}_{r}$ and ${\overline{\bf{Q}}}_{t}$ can be designed to full column rank, it is possible to obtain any $KN$-row block matrix of $\widetilde{\bf H}_w$ that is uncorrelated, thereby allowing the rank of the channel matrix to equal $KN$. When $KN> \min\left(\tilde{d}_t,\tilde{d}_r\right)$, any $KN$ row vectors of $\widetilde{\bf H}_w$ will be correlated. The agile adjustment of phase response can achieve $\min\left(\tilde{d}_t,\tilde{d}_r\right)$ uncorrelated row vectors, resulting in a channel rank of $\min\left(\tilde{d}_t,\tilde{d}_r\right)$. In this case, agilely adjusting phase responses of DARISA MIMO can improve the rank of the composite  spatial-temporal channel such that $\textrm{rank}\left(\overline{\bf H}_C\right)= \min\left(KN,\tilde{d}_t,\tilde{d}_r\right)$.
		
		2) When $M< KN$, the rank of the composite spatial-temporal channel matrix is limited by $M$, with the channel rank determined by the $M$ rows of the block matrix of $\widetilde{\bf H}_w$. This leads to $\textrm{rank}\left(\overline{\bf H}_C\right)= \min\left(M,\tilde{d}_t,\tilde{d}_r\right)$ based on the similar analysis of the above case. In this case, agilely adjusting phase responses cannot improve the rank of the composite spatial-temporal channel.
		
		According to the above analysis, we obtain $\textrm{rank}({\overline{\bf H}}_C)= \min\left(KN,M,\tilde{d}_t,\tilde{d}_r\right)$.
		
			\section{ Proof of Lemma 2} \label{proof_lemma2}
			
			By denoting $s=\Psi\left(\overline{\mathbf{H}}_C\right), a=\textnormal{SNR}$,  the channel capacity (\ref{EDoF-capacity}) is equivalently expressed as
			\begin{align}
				\tag{A-15}	\label{dof-fuction}
				f(s)=\frac{s}{\ln 2} \ln\left(1+\frac{a}{s}\right),s\geq 0.
			\end{align}
			
			\quad We will analyze the monotonicity of the function $f(s)=s \ln \left(1+\frac{a}{s}\right)$ by ignoring the constant $\frac{1}{\ln 2}$. The first order derivative is exprerssed as 
			\begin{align}
				\tag{A-16}	\label{derivative}
				f^{\prime}(s)=\ln \left(1+\frac{a}{s}\right)-\frac{a}{s+a},
			\end{align}
			then we determine the positivity/negativity of the above equation.
			
			 The proof is along the following steps: 1) prove that the first-order derivative is nonnegative at both 0 and $+\infty$; 2) prove that the first-order derivative is monotonically increasing.
			
			Step 1:  When $s \rightarrow 0^{+}$:
			$\ln \left(1+\frac{a}{s}\right) \rightarrow+\infty$, and $\frac{a}{s+a} \rightarrow 1$, so $f^{\prime}(s) \rightarrow+\infty$, the derivative is positive.
			
			 When $s \rightarrow+\infty$:
			Approximate $\ln \left(1+\frac{a}{s}\right) \approx \frac{a}{s}-\frac{a^2}{2 s^2}$ with Taylor expansion and $\frac{a}{s+a} \approx \frac{a}{s}-\frac{a^2}{s^2}$. Substituting $s \rightarrow+\infty$ into (\ref{derivative}), the derivative is given by:
			$$
			f^{\prime}(s) \approx\left(\frac{a}{s}-\frac{a^2}{2 s^2}\right)-\left(\frac{a}{s}-\frac{a^2}{s^2}\right) = \frac{a^2}{2 s^2}>0
			$$
			Therefore, the derivative tends to zero but remains positive.
			
			Step 2: Let $z=\frac{a}{s} \quad(z>0)$ , the derivative in (\ref{derivative}) can be rewritten as:
			\begin{align}
				\notag
				f^{\prime}(s)=\ln (1+z)-\frac{z}{1+z},
			\end{align}
			define the function $$ g(z)=\ln (1+z)-\frac{z}{1+z}, $$ and its derivative is expressed as:
			$$
			g^{\prime}(z)=\frac{1}{1+z}-\frac{1}{(1+z)^2}=\frac{z}{(1+z)^2}>0 \quad(z>0),
			$$
			thus, $g(z)$ is monotonically increasing at $z>0$.  Since $g(0)=0, x \rightarrow+\infty$ and $g(z)>0$ holds for all $z>0$, $f^{\prime}(s)>0$ is valid for all $s>0$.
			
			Therefore, the function (\ref{dof-fuction}) always has a positive derivative at $s>0$, so the function is monotonically increasing, i.e., the channel capacity monotonically increases with EDoF. 
		
		\section{ Proof of Lemma 3} \label{proof_lemma3}
		According to (\ref{final_siso_channel}) and (\ref{spatial-temporal-MIMO_channel}), the channel response of $kn$-th row and $m$-th column element in $\overline{\mathbf{H}}_C\in \mathbb{C}^{KN\times M}$ is expressed as $\overline{\mathbf{H}}_C(kn,m)
		=\frac{1}{\sqrt{\tilde L}}\sum\limits_{i=1}^{N_r}\sum\limits_{j=1}^{N_t}\sum\limits_{\tilde l=1}^{\tilde L} {\widetilde q}_{r,n}^{i}(t_k){\widetilde{\bf a}}_{r,n}^{\tilde{l},iH}{\widetilde{\bf H}}_{a}^{\tilde l} {\widetilde{\bf a}}_{t,m}^{{\tilde l},j}{\widetilde q}_{t,m}^{j}(t_k)$, where
		$\widetilde{\bf a}_{r,n}^{{\tilde l},i}\in \mathbb{C}^{{\tilde d}_r^{\tilde l}\times 1}$ represents the receive response of the $\tilde l$-th arrival cluster at $k$-th agility duration of the $i$-th element of $n$-th receive-DARISA, $\widetilde{\bf a}_{t,m}^{{\tilde l},j}\in \mathbb{C}^{{\tilde d}_t^{\tilde l}\times 1}$ represents the transmit response of the $\tilde l$-th departure cluster at $k$-th agility duration of the $j$-th element of $m$-th transmit-DARISA, $\widetilde{\bf H}_{a}^{\tilde l}\in \mathbb{C}^{d_r^{\tilde l}\times d_t^{\tilde l}}$ represents the scattering response of the $\tilde l$-th cluster, ${\widetilde q}_{r}^{n,i}(t_k)=e^{\jmath\varphi_{r}^{n,i}(t_k)}$ $\left({\widetilde q}_{t}^{m,j}(t_k)=e^{\jmath\varphi_{t}^{m,j}(t_k)}\right)$ is phase response of $i$-th ($j$-th) element of $n$-th ($m$-th) receive  (transmit)  DARISA at $k$-th agility duration.
		
		Since phase responses of all elements can be adjusted continuously within $(0,2\pi]$, i.e., $\varphi_{r}^{n,i}(t_k)\in  (0,2\pi],\varphi_{t}^{m,j}(t_k)\in  (0,2\pi]$, by adjusting ${\widetilde q}_{r,n}^{i}(t_k)$, we always obtain
		\begin{align}
			\notag
			{\widetilde q}_{r,n}^{i}(t_k){\widetilde{\bf a}}_{r,n}^{\tilde{l},iH}(t_k) {\widetilde{\bf H}}_{a}^{\tilde l} {\widetilde{\bf a}}_{t,m}^{{\tilde l},j}(t_k){\widetilde q}_{t,m}^{j}(t_k) = {\widetilde q}_{r,n}^{i}(t_k){\widetilde{\bf a}}_{r,n}^{\tilde{l},iH}{\widetilde{\bf H}}_{a}^{\tilde l} {\widetilde{\bf a}}_{t,m}^{{\tilde l},j}.
		\end{align}
		This result implies that optimizing the phases of both $\overline{\bf{Q}}_{t}$ and $\overline{\bf{Q}}_{r}$ is equivalently to optimizing the phase of $\overline{\bf{Q}}_{r}$ only.
		
		\section{ Proof of Theorem 2} \label{proof of thm2}
		First, we analyze the range of $\zeta^{opt}$, which must satisfy
		\begin{align}
			\tag{A-17}\label{t}
			\zeta^{opt}=\max_{\varphi_{r}^{n,i}(t_k)} {\textrm{Tr}\left(\overline{\mathbf{H}}_C\overline{\mathbf{H}}_C^H\right)}\big{/}{\big{|}\big{|}\overline{\mathbf{H}}_C\overline{\mathbf{H}}_C^H\big{|}\big{|}_F}.
		\end{align}
		When the rank of $\overline{\mathbf{H}}_C$ is $1$, $\zeta^{opt}$ reaches its minimum value of $1$; when $\overline{\mathbf{H}}_C$ is full rank and all eigenvalues are equal, $\zeta^{opt}$ attains maximum value of $\textrm{rank}\left(\overline{\mathbf{H}}_C\right)$. Therefore, the range of $\zeta^{opt}$ is
		\begin{align}
			\tag{A-18}
			1\leq \zeta^{opt} \leq \sqrt{\textrm{rank}\left(\overline{\mathbf{H}}_C\right)}.
		\end{align}
		According to Theorem 1, the rank of ${\overline{\bf H}}_C$ is given by $\textrm{rank}\left({\overline{\bf H}}_C\right)= \min\left(KN,M,\tilde{d}_t,\tilde{d}_r\right)$.
		
		Therefore, we only need to show that there exists a unique zero root for the objective function in (\ref{edof2}) within a specified range $\zeta \in \left[1, \sqrt{\textrm{rank}\left(\overline{\mathbf{H}}_C\right)}\right]$.
		
		When $\zeta=1$, based on the square inequality, we have $\left(\sum_{a=1}^{\textrm{rank}\left(\overline{\mathbf{H}}_C\right)}\lambda_a^2\right)^2\geq \sum_{a=1}^{\textrm{rank}\left(\overline{\mathbf{H}}_C\right)}\lambda_a^4$.
		The equation holds if and only if the rank of $\overline{\mathbf{H}}_C$ is equal to $1$. For $\forall {\bf E}$, the following inequality holds:
		\begin{align}
			\notag
			\textrm{Tr}\left({\bf C}{\bf E}\right)-{\big{|}\big{|}{\bf C}{\bf E}\big{|}\big{|}_F}=\sum_{a=1}^{\textrm{rank}\left(\overline{\mathbf{H}}_C\right)}\lambda_a^2-\sqrt{\sum_{a=1}^{\textrm{rank}\left(\overline{\mathbf{H}}_C\right)}\lambda_a^4}\geq 0.
		\end{align}
		Thus, we obtain
		\begin{align}
			\notag
			F(1)=\max_{{\bf E}}\left\{\textrm{Tr}\left({\bf C}{\bf E}\right)-{\big{|}\big{|}{\bf C}{\bf E}\big{|}\big{|}_F}\right\}\geq 0.
		\end{align}
		
		When $\zeta=\sqrt{\textrm{rank}\left(\overline{\mathbf{H}}_C\right)}$, since  $\left(\sum_{a=1}^{\textrm{rank}\left(\overline{\mathbf{H}}_C\right)}\lambda_a^2\right)^2\leq \textrm{rank}\left(\overline{\mathbf{H}}_C\right) \left( \sum_{a=1}^{\textrm{rank}\left(\overline{\mathbf{H}}_C\right)}\lambda_a^4\right)$,
		the equation holds if and only if all eigenvalues of $\overline{\mathbf{H}}_C$ are equal. For any $\forall {\bf E}$, the following inequality holds:
		\begin{align}
			\notag
			&\textrm{Tr}\left({\bf C}{\bf E}\right)-\sqrt{\textrm{rank}\left(\overline{\mathbf{H}}_C\right)}{\big{|}\big{|}{\bf C}{\bf E}\big{|}\big{|}_F}\\
			\notag
			&=\sum_{a=1}^{\textrm{rank}\left(\overline{\mathbf{H}}_C\right)}\lambda_a^2-\sqrt{\textrm{rank}\left(\overline{\mathbf{H}}_C\right)\sum_{a=1}^{\textrm{rank}\left(\overline{\mathbf{H}}_C\right)}\lambda_a^4}\leq 0.
		\end{align}
		
		Therefore, denoting $\zeta^{u}\triangleq\sqrt{\textrm{rank}\left(\overline{\mathbf{H}}_C\right)}$, we obtain
		\begin{align}
			\notag
			&F\left(\zeta^{u}\right)=\max_{{\bf E}}\left\{\textrm{Tr}\left({\bf C}{\bf E}\right)-\zeta^{u}{\big{|}\big{|}{\bf C}{\bf E}\big{|}\big{|}_F}\right\}\leq 0;
		\end{align}
		if and only if all eigenvalues of ${\overline{\bf H}}_C$ are equal, $F\left(\zeta^u\right)=0$.
		
		Based on the Lemma 4, the optimal objective function of the problem (\ref{edof2}) is strictly monotonically decreasing with respect to $\zeta$.
		For $\zeta\in \left[1,\sqrt{\textrm{rank}\left(\overline{\mathbf{H}}_C\right)}\right]$, we have $F(1)\geq 0$ and $F\left(\sqrt{\textrm{rank}\left(\overline{\mathbf{H}}_C\right)}\right)\leq 0$. If $\textrm{rank}\left(\overline{\mathbf{H}}_C\right)=1$, $F(1)=F\left(\sqrt{\textrm{rank}\left(\overline{\mathbf{H}}_C\right)}\right)=0$; If $\textrm{rank}\left(\overline{\mathbf{H}}_C\right)>1$ and all eigenvalue values of ${\overline{\bf H}}_C$ are equal, $F(1)> 0$ and $F\left(\sqrt{\textrm{rank}\left(\overline{\mathbf{H}}_C\right)}\right)=0$; Otherwise, $F(1)> 0$ and $F\left(\sqrt{\textrm{rank}\left(\overline{\mathbf{H}}_C\right)}\right)<0$.
		Therefore, there must exist a unique $\zeta^{opt}$ such that the optimal objective function of problem (\ref{edof2}) is equal to $0$. }

\end{document}